\documentclass[12pt]{amsart}

\usepackage[a4paper, 
            left=1.in,
            right=1in,
            top=1in,
            bottom=1in, 
            footskip=.25in]{geometry}

\usepackage[utf8]{inputenc}
\usepackage{graphicx}
\usepackage{amssymb} 
\usepackage[colorlinks = true,
            linkcolor = blue,
            urlcolor  = blue,
            citecolor = blue,
            anchorcolor = blue]{hyperref}

\usepackage{enumerate}
\usepackage[shortlabels]{enumitem}
\usepackage{thmtools} 
\usepackage{thm-restate}

\usepackage{microtype}

\usepackage{xcolor}
\usepackage{cite}
\usepackage{cleveref}
\usepackage{centernot}  
\usepackage{xspace}  
\usepackage{subcaption}
\usepackage{pdfcomment}
\usepackage{todonotes}
\usepackage{mathrsfs}
\usepackage{comment}

\usepackage{tikz}
\usetikzlibrary{matrix}

\usepackage{caption}
\theoremstyle{plain}
\newtheorem{thm}{}[section]
\newtheorem{lemma}[thm]{Lemma}
\newtheorem{proposition}[thm]{Proposition}

\newtheorem{theorem}[thm]{Theorem}
\newtheorem{corollary}[thm]{Corollary}
\theoremstyle{remark}

\theoremstyle{definition}
\newtheorem{problem}{Problem}
\newtheorem{definition}{Definition}

\usepackage[scr=boondoxo]{mathalpha}  
 
\newcommand{\structure}[1]{\mathbb{#1}}

\newcommand{\family}[1]{\mathcal{#1}}
\newcommand{\forb}[1]{\operatorname{Forb}_h(#1)}
\newcommand{\dom}[1]{\mathrm{dom}(#1)}

\newcommand{\Aut}[0]{\mathrm{Aut}}

\DeclareMathOperator\dist{dist}
\DeclareMathOperator\CSP{CSP}
\DeclareMathOperator\Ext{Ext}
\DeclareMathOperator\Col{Col}

\newcommand\ignore[1]{}

\AtBeginDocument{\fontsize{12pt}{15pt}\selectfont}

\begin{document}

\title{Edge-coloring problems with forbidden patterns and planted colors}
 
\thanks{Alexey Barsukov is funded by the European Union (ERC, POCOCOP, 101071674). Views and opinions expressed are however those of the author(s) only and do not necessarily reflect those of the European Union or the European Research Council Executive Agency. Neither the European Union nor the granting authority can be held responsible for them.}
 
\author{Alexey Barsukov}
\author{Antoine Mottet}
\author{Davide Perinti}

\address{Faculty of Mathematics and Physics, Charles University, Prague, Czechia}
\email{alexey.barsukov@matfyz.cuni.cz}
 
\address{Hamburg University of Technology, Research Group on Theoretical Computer Science, Germany}
\email{$\{$antoine.mottet,davide.perinti$\}$@tuhh.de}
 
\begin{abstract}
    Edge-coloring problems with forbidden patterns are decision problems asking to find an edge-coloring of the input graph which avoids a homomorphism from a fixed forbidden family of edge-colored graphs.
    In the precolored version of these problems, some of the edges of the input graph are already colored, and the goal is to find an extension of this coloring which omits a homomorphism from a forbidden graph.
    
    The existence of a complexity classification for such problems is an open question of Bienvenu, ten Cate, Lutz, and Wolter (ACM TODS'14) and we answer it for certain forbidden families consisting of odd cycles and cliques.
    The proof consists of two main stages.
    First, we combine the techniques from infinite constraint satisfaction and finite Ramsey theory in order to show that the edge-coloring problem is poly-time equivalent to its precolored version.
    After that, we show that the precolored version is poly-time equivalent to a finite constraint satisfaction problem, which has a P vs.\ NP-complete dichotomy by the seminal results of Bulatov (FOCS'17)  and Zhuk (FOCS'17).
\end{abstract}

\maketitle


\section{Introduction}
\label{sec:intro}

The following is a well-known instance of Ramsey's theorem: it is impossible to color the edges of the complete graph on $R_2(3)=6$ vertices with $c=2$ colors in a way that avoids a monochromatic complete graph $\structure{K}_n$ of size $n=3$.
More generally, Ramsey's theorem states that for every $c$ and $n$, there is a number $R_c(n)$ such that no complete graph with at least $R_c(n)$ vertices can be colored with $c$ colors while avoiding monochromatic copies of $\structure{K}_n$.
The exact values of $R_c(n)$ are unknown even for moderate values of $c$ and $n$, and determining them computationally is hard~\cite{GrahamSpencer1990Ramsey}.
Generalizing cliques to arbitrary finite edge-colored graphs leads to the following computational problem $\Col(\mathcal F)$.
\begin{problem}[$\Col(\mathcal F)$]
    Fix $c\in \mathbb N$ -- number of colors and $\mathcal F\subset \big\{(\mathbb F, f)\ \big|\ \mathbb F \text{ - graph},\ f\colon E^\mathbb F\to [c]\big\}$ -- finite family of edge-colored graphs.\\
    \textbf{Input:} A graph $\mathbb G$.\\
    \textbf{Question:} Is there an edge-coloring $s\colon E^\mathbb G\to[c]$ such that for any $(\mathbb F,f)\in \mathcal F$ and any homomorphism $h\colon \mathbb F\to \mathbb G$, we never have $f = s\circ h$?
\end{problem}

The edge-colored graphs from the family $\mathcal F$ are called \emph{forbidden patterns}.
If an edge-coloring $s$ witnesses the positive answer to the question, then it is called \emph{$\mathcal F$-free}.
Edge-coloring problems appeared in theoretical computer science long ago: e.g., Garey and Johnson's list of NP-complete problems~\cite{GareyJohnson} has a problem \textsc{No-Mono-Tri} which asks for an edge-coloring in 2 colors without monochromatic triangles. 
In our setting, this problem is denoted by $\Col\big(\mathbb K_3^{(1)},\mathbb K_3^{(2)}\big)$, where $\structure{K}_n^{(i)}$ denotes a complete graph on $n$ vertices all of whose edges have color $i\in\{1,2\}$.
Later, Burr~\cite{Burr} proved that $\Col\bigl(\structure{K}_n^{(1)},\structure{K}_n^{(2)}\bigr)$ is also NP-complete for all $n> 3$.
Such Ramsey-type problems have been considered for decades in the variants where the graphs $\structure{K}_n$ are replaced by other graphs, mostly in the case where all obstructions are monochromatic and where every color has a single monochromatic obstruction.

\subsection{The logical and computational relevance of Ramsey-type problems}
Besides their combinatorial interest, problems of the form $\Col(\family F)$ have appeared in connection with results in the area of ontology-mediated query answering~\cite{OBDA}.
Consider the Guarded Fragment of first-order logic (GFO).
In~\cite[Theorems 3.17 and 4.2]{OBDA}, it is established that answering a union of conjunctive queries mediated by an ontology written in GFO can be phrased as an instance of $\Col(\family F)$ for a suitably chosen family of obstructions $\family F$ depending on the query to be answered.
However, those obstructions are not necessarily monochromatic as is the case in the Ramsey setting, and obstructions are forbidden by \emph{homomorphisms}.
Moreover, such problems are exactly those problems that can be described in the fragment of existential second-order logic called \emph{Guarded Monotone SNP} (GMSNP).
This logic was independently introduced by Madelaine~\cite{MadelaineMMSNP2} and Bienvenu, ten Cate, Lutz, and Wolter~\cite{OBDA}.
GMSNP is an extension of Feder and Vardi's logic \emph{Monotone Monadic SNP without inequality} (MMSNP)~\cite{FederVardi}, and it relates to MMSNP in the same way as the logic MSO$_2$ of Courcelle and Engelfriet~\cite{Courcelle} relates to monadic second-order logic. 
Namely, second-order quantifiers of arity greater than $1$ are allowed in the syntax, as long as their range is guarded by the input predicates from the signature.
Then, for every GMSNP sentence $\Phi$, the class of finite graphs satisfying $\Phi$ coincides with the class of yes-instances of the problem $\Col(\family F)$, where $\family F$ is a finite set of edge-colored graphs as above.

From the perspective of descriptive complexity it is natural to ask, for the problems that can be described in GMSNP, what complexity they can have.
An oft-mentioned open problem is whether there is a \emph{complexity dichotomy}, i.e., if every such problem is either solvable in poly-time or NP-complete~\cite{OBDA,MadelaineMMSNP2,MMSNPjournal}.
This happens to be true for MMSNP~\cite{FederVardi}, and the dichotomy has been proved for some other subclasses of GMSNP~\cite{Bitter,Santiago,Feller}.
Conversely, it is known that for variants of such coloring problems, for example, problems asking  for a vertex-coloring omitting certain vertex-colored subgraphs~\cite[Theorem 4.1]{KunNesetril} can encode every problem in NP (modulo poly-time equivalence) and, by Ladner's theorem, have no complexity dichotomy unless P = NP.

Problems of the form $\Col(\family F)$ can be seen as finite unions of \emph{constraint satisfaction problems} (CSPs) with countably infinite templates satisfying certain model-theoretic properties (they are the so-called \emph{reducts of finitely bounded homogeneous structures})~\cite{ASNP}.
A complexity dichotomy for such problems has been conjectured and proved to be true in many particular cases (see, e.g., the introduction of~\cite{SmoothApproximations}).
Thus, the complexity dichotomy for $\Col(\mathcal F)$ would be a consequence of the dichotomy conjectured for the CSPs mentioned here, and $\Col(\mathcal F)$ gives a rich class of examples on which to put the CSP dichotomy conjecture to the test.

\subsection{The complexity of vertex-coloring problems}
In the version (also known as MMSNP) of the problem $\Col(\mathcal F)$ where vertex-colorings are considered instead of edge-colorings, the complexity of the problems is well-understood: this class of problems admits a P vs.\ NP-complete dichotomy.
Feder and Vardi initially showed in their seminal paper~\cite{FederVardi} that such problems are equivalent under randomized poly-time reductions to finite-domain CSPs.
While the reduction from MMSNP to CSP is immediate, the backwards reduction relies on the so-called \emph{Sparse Incomparability Lemma} which roughly states that whenever a finite constraint satisfaction problem has a no-instance, then it has no instances of arbitrarily large girth.
The construction of such a high girth instance was originally probabilistic and was later derandomized by Kun~\cite{Kun}.
Since Bulatov~\cite{Bulatov} and Zhuk~\cite{Zhuk}  proved that finite-domain CSPs obey a complexity dichotomy, this implies a complexity dichotomy for vertex-coloring problems by the reduction above.
It has been an open problem in the area whether such a proof strategy relying on the sparse incomparability lemma could work for edge-colorings. However, this problem remains open, despite nearly two decades of effort since Madelaine's introduction of the logic.

More recently, Bodirsky, Madelaine, and Mottet~\cite{MMSNPjournal} obtained a new proof of the complexity dichotomy for MMSNP, relying on the algebraic approach to \emph{infinite-domain} constraint satisfaction rather than the sparse incomparability lemma. 
While vertex-coloring problems are not expressible as finite-domain CSPs, it was shown in~\cite{BodirskyDalmau} that for every family $\mathcal F$ of connected structures, there is a countably infinite vertex-colored structure $\mathbb C_\mathcal F$ such that a vertex-colored structure maps homomorphically to $\mathbb C_\mathcal F$ if and only if it is $\mathcal F$-free.
Their dichotomy proof consists of two main steps.
The first is a reduction of the dichotomy question to the so-called \emph{precolored} setting, where some vertices of the structure given as input to the problem are already colored, and the problem is to decide the existence of an extension of this coloring to an $\mathcal F$-free vertex-coloring of the whole structure.
This problem, denoted by~$\Ext(\family F)$ in the following, can in principle be harder than the non-precolored problem; however, it is shown in~\cite{MMSNPjournal} that without loss of generality, the vertex-coloring versions of the problems $\Ext(\family F)$ and $\Col(\family F)$ are computationally equivalent.
The second step is the poly-time equivalence between $\Ext(\family F)$ and the CSP of the finite structure $\mathbb D_\mathcal F$ that was used in the original reduction of Feder and Vardi~\cite{FederVardi} as well.
To bypass the use of the sparse incomparability lemma, the equivalence is obtained indirectly in the following way: if the CSP of $\mathbb D_{\mathcal F}$ is NP-hard, then an algebraic witness of this hardness exists.
The major result of Bodirsky, Madelaine, and Mottet~\cite{MMSNPjournal} is that such an algebraic witness can be lifted to prove the NP-hardness of the problem $\Ext(\family F)$.
Thus, the two problems are equivalent under poly-time reductions.
Since then, this proof strategy has been generalized in a variety of other contexts, developing what is now known as the theory of smooth approximations~\cite{SmoothApproximations,MMSNPwidth,SmoothApproximations3,Feller,Bitter}.
Although this strategy can in principle be applied to the case of edge-coloring problems, both subproblems above (the reduction to the precolored setting, and the lifting of algebraic NP-hardness witnesses) remain hitherto unsolved.

\subsection{Contributions}

This paper provides new methods for investigating the complexity of edge-coloring problems.
They take inspiration from the tools that were used to study the vertex-coloring problems from both aforementioned papers~\cite{MMSNPjournal,FederVardi}.
Our first contribution is a complexity dichotomy for many natural and interesting cases of families $\family F$ consisting of odd cycles and cliques.

\begin{restatable}{theorem}{mainapplication}\label[theorem]{thm:main}
    Let $\mathcal F$ be a finite set of monochromatic odd cycles or cliques,
    and let $\structure S$ be an odd cycle or a clique.
    Then the problem $\Col(\mathcal F)$ restricted to $\structure S$-free instances is either trivial or NP-complete.
\end{restatable}

Note that $\Col(\mathcal F)$ is trivial if and only if $\family F$ does not contain a monochromatic obstruction for some color (in this case, we just assign to every edge the missing color). 

The hardness of the problems from the scope of \Cref{thm:main} allows us to extend this complexity classification to much larger classes of colored obstructions.

\begin{corollary}\label{corollary:main-cons}
    Let $\mathcal M$ be a finite set of monochromatic odd cycles or cliques that contains a monochromatic graph for each color $i\in [c]$, and $\mathcal N$ be a set of colored graphs such that there exists an odd cycle or a clique $\structure S$ with the property that 
    $
        (\structure G,\chi) \in \mathcal{M} \Rightarrow \structure S \not\rightarrow \structure G \text{ and } (\structure G,\chi) \in \mathcal{N} \Rightarrow \structure S \rightarrow \structure G.
    $
    Then $\Col(\family M\cup\family N)$ is NP-complete.
\end{corollary}

Most of the known results concerning edge-coloring problems with forbidden patterns are in fact restricted to the case where the obstructions are monochromatic cliques (or at least 3-connected graphs), or the case where the obstructions are monochromatic cycles.
Ours is the only complexity result with mixed obstructions, also considering non-monochromatic obstructions.
As an example of a consequence of our result, we obtain  the following refinement of the theorem of Burr~\cite{Burr} that $\Col\big(\structure{K}^{(1)}_n,\structure{K}^{(2)}_n\big)$ is NP-hard for all $n\geq 3$:
\begin{corollary}
    For  $n\geq 3$, the monochromatic $\structure{K}_n$ problem is NP-hard on $\structure{K}_{n+1}$-free graphs.
\end{corollary}

The most technical part of the proof of~\Cref{thm:main} is the reduction from $\Ext(\mathcal F)$ to $\Col(\mathcal F)$ in \Cref{sec:infinite_csp,section:coldet}.
It combines the model-theoretic approach of~\cite{MMSNPjournal} with the Ramsey-theoretic results of~\cite{NESETRIL1976243,HAN2015457}.
We show that, for families $\family F$ as in~\Cref{thm:main}, one can construct color-determiners and other types of gadgets introduced in Ramsey theory under the name of \emph{signal senders} by Burr, Erd\H{o}s, and Lov\'asz~\cite{BurrErdosLovasz}.
The existence of such gadgets has been confirmed in numerous cases involving monochromatic obstructions and very simple graphs (see~\cite{BurrNesetrilRoedl,RoedlSiggers,BoyadzhiyskaLesgourgues} for papers dedicated to the construction of such gadgets).
Our result is also a contribution towards the study of such gadgets in the setting of homomorphically forbidden graphs, and proving their existence is the core of the present paper.

The second result of our paper concerns the inexpressibility of edge-coloring problems as vertex-coloring problems (i.e., those problems corresponding to the logic MMSNP).
The fact that the former class is strictly more expressive than the latter is already known~\cite[Cor.~4.4]{OBDA}.
The property witnessing the inexpressibility was defined over a signature with 3 relation symbols of arities 1, 1, and 3, and the criterion of being in MMSNP was a version of a pebble game.
We give a new proof of this separation with very simple counterexamples.
In particular, the \textsc{Edge-No-Mono-Tri} problem satisfies the premise of the following theorem by~\cite[Theorem 1.1]{RoedlSiggers} and therefore cannot be expressed as a vertex-coloring problem, which answers the question of Madelaine~\cite[Remark 4.19]{MadelaineMMSNP2}.
\begin{restatable}{theorem}{separation}\label{separation_gmsnp}
    Let $\family E$ be a finite set of edge-colored graphs with the following properties:
    \begin{itemize}
        \item the class of $\family E$-free edge-colored graphs is closed under amalgamation over vertices,
        \item for every $K\geq 1$, there exists a graph $\structure{G}$ that does not admit an $\family E$-free coloring but every $K$-element subgraph of $\structure{G}$ does.
    \end{itemize}
    Then there is no finite set $\family F$ of vertex-colored graphs such that $\Col(\family F)=\Col(\family E)$.
\end{restatable}

The first assumption can be made without loss of generality, as for any class $\family E$ of edge-colored graphs, there exists $\family E'$ such that $\Col(\family E)=\Col(\family E')$ and such that the class of $\family E'$-free edge-colored graphs is closed under amalgamation over vertices.
The second assumption in~\Cref{separation_gmsnp} is mild, as it is equivalent to saying that the class of $\family E$-free colorable graphs is not definable by a first-order sentence~\cite{Rossman}.
This condition is clearly necessary to obtain a class that is definable in GMSNP and not in MMSNP, since any such coloring problem can be phrased with obstructions without using \emph{any} colors.
Note that this condition is decidable, which gives an algorithm for the question to decide whether a given GMSNP sentence corresponding to a graph edge-coloring problem is equivalent to an MMSNP sentence.

\section{Notations and Definitions}
\label{sec:prelims}

Let $\tau=\{R_1,\ldots,R_s\}$ be a set of relation symbols.
Each $R_i\in\tau$ is assigned some number $k_i\in\mathbb{N}$ called the \emph{arity} of $R_i$.
A \emph{$\tau$-structure} $\structure{A}$ is a tuple $(A; R_1^\structure{A},\ldots,R_s^\structure{A})$, where $A$ is a set, called \emph{domain}, and each $R_i^\structure{A}$, called \emph{relation}, is a subset of $A^{k_i}$.

For $\tau$-structures $\structure{A}$ and $\structure{B}$, a mapping $h\colon A\to B$ is a \emph{homomorphism} if, for every $R\in\tau$ and every $\bar a\in R^\structure{A}$, the tuple $h(\bar a)$ obtained by applying $h$ componentwise is contained in $R^\structure{B}$.
The decision problem $\CSP(\structure{A})$ asks whether, for a finite  $\structure{X}$, there is a homomorphism $\structure{X}\to\structure{A}$.

If not explicitly stated otherwise, $[c]$ denotes the set of edge colors.
An \emph{edge-colored graph} is a pair $(\structure{G},\chi^\mathbb G)$, where $\chi^\mathbb G\colon E^\structure{G}\to[c]$.
We write $(\mathbb G,\chi^\mathbb G)\to(\mathbb H,\chi^\mathbb H)$ if there is a homomorphism $h\colon \mathbb G\to\mathbb H$ such that $\chi^\mathbb H\circ h = \chi^\mathbb G$.
Homomorphisms define a preorder relation $\to$ which is the linear order 
$\dots\to\structure C_7\to\structure C_5\to\structure C_3=\structure K_3\to\structure K_4\to\structure K_5\to\dots$
when restricted to odd cycles (denoted $\mathbb C_\ell$ in the following) and cliques (denoted $\mathbb K_\ell$).

For a family $\family{F}$ of edge-colored graphs,  $(\structure{G},\chi)$ is \emph{$\family{F}$-free} if $(\structure{F},\phi)\not\to(\structure{G},\chi)$, for every $(\structure{F},\phi)\in\family{F}$.
Let $\forb{\family{F}}$ be the set of all finite $\family{F}$-free edge-colored graphs.
The decision problem $\Col(\family{F})$ asks whether, for an input graph $\structure{X}$, there exists $\chi^\mathbb X\colon E^\structure{X}\to[c]$ such that $(\structure{X},\chi^\mathbb X)$ is $\family{F}$-free; the set of available colors is always implicitly assumed to precisely consist of the colors appearing in $\family F$.

A \emph{partially edge-colored graph} is a pair $(\structure{G},\xi^\mathbb G)$, where $\xi^\mathbb G\colon S\to[c]$, for a subset of edges $S\subseteq E^\structure{G}$.
For $S\subseteq E^\structure{G}$, a mapping $\chi^\mathbb G\colon E^\structure{G}\to [c]$ is an \emph{extension} of $\xi^\mathbb G\colon S\to [c]$ if $\chi^\mathbb G\restriction_S = \xi^\mathbb G$.
The decision problem $\Ext(\family{F})$ asks whether, for a partially edge-colored graph $(\structure{X},\xi^\mathbb X)$, there exists $\chi^\mathbb X$ extending $\xi^\mathbb X$ such that $(\structure{X},\chi^\mathbb X)$ is $\family{F}$-free.

For two vertices $v,w \in G$ of a graph $\structure{G}$, the \emph{distance} between $v$ and $w$ is the number of edges in a shortest path connecting $v$ to $w$.
The distance $\dist(e,e')$ between two edges $e, e'\in E^\structure{G}$ is defined as the minimum distance between the vertices incident to $e$ and $e'$, respectively.
For a set of edges $S \subseteq E^\structure{G}$ and an edge $e\in E^\structure{G}$, the distance $\dist(e, S)$ is defined as $\min_{e'\in S}\dist(e,e')$.

A graph is \emph{connected} if it is not the disjoint union of two other graphs.
Let $\structure G$ and $\structure H$ be graphs with distinguished pairwise disjoint edges $e_1,\dots,e_k\in E^\structure G$ and $f_1,\dots,f_k\in E^\structure H$, respectively.
We denote the graph obtained by taking the disjoint union of $\structure G$ and $\structure H$ and by gluing $e_i$ and $f_i$ for all $i\in\{1,\dots,k\}$ by $\structure G\oplus\structure H$. 
In the case where $\mathbb G$ and $\mathbb H$ are partially colored with $\xi^\mathbb G$ and $\xi^\mathbb H$ such that $\xi^\mathbb G(e_i)=\xi^\mathbb H(f_i)$ for every $i$, we write $\xi^\mathbb G\oplus \xi^\mathbb H$ for the partial coloring of $\structure G\oplus\structure H$.
While the notation omits the distinguished edges, they are clear from context in the places where we use this notation; the notation also omits how the endpoints of $e_i$ and $f_i$ are glued, all our results do not depend on this choice unless stated explicitly.

Given a graph $\structure G$, edges $e_1,\ldots, e_k\in E^\mathbb G$, and $i\in [c]$, let $(i\setminus e_1\ldots e_k)$ be the partial coloring of $\structure G$ which assigns the color $i$ to all edges except for $e_1,\ldots,e_k$.
Let $\structure G$ be a graph, $\xi^\structure G$ be a partial coloring of its edges, and $e_1,\dots,e_k \in E^\mathbb G\setminus\dom{\xi^\structure G}$ be pairwise disjoint.
We say that $(\structure{G}, \xi^\structure G, e_1,\dots,e_k)$ is \emph{$\family F$-safe} if:
\begin{itemize}
    \item there exists an $\family{F}$-free extension of $\xi^\mathbb G$,
    \item for every $\family F$-free edge-colored graph $(\structure H,\gamma^{\structure H})$ and pairwise disjoint $f_1,\dots,f_k\in E^\structure H$, if there is an extension $\gamma^\structure G$ of $\xi^\structure G$ satisfying $\gamma^\structure G(e_i)=\gamma^{\structure H}(f_i)$ for all $i\in [c]$, then $(\structure G \oplus \structure H, \gamma^\structure G\oplus \gamma^{\structure H})$ is $\mathcal F$-free, where $\structure G \oplus \structure H$ is obtained by identifying $e_i$ with $f_i$ for all $i$.
\end{itemize}

Given $\family F$ as in the scope of~\Cref{thm:main}, we indicate via $\mathbb F_i$ the $\rightarrow$-minimum monochromatic graph in $\family F$ colored in $i$.
\begin{proposition}\label{easyobs}
Let $\mathcal F$ be in the scope of \Cref{thm:main}. Then the following properties hold:
    \begin{enumerate}[(1)]
        \item\label{itm:safenessamalgam} Let $(\structure G,\xi^\structure G, e_1,\dots,e_k),(\structure H,\xi^\structure H,f_1, \dots, f_k)$ be $\family F$-safe.
        Suppose that there exist an $\family F$-free extension $\chi^{\structure G}$ of $\xi^{\structure G}$ and an $\family F$-free extension $\chi^{\structure H}$ of $\xi^{\structure H}$ such that $\chi^{\structure G}(e_j)=\chi^{\structure H}(f_j)$ for all $j\in\{1,\dots,k\}$.
        Then $(\structure G \oplus \mathbb H,\xi^\structure G \oplus \xi^\structure H, e_1\equiv f_1,\ldots,e_k \equiv f_k)$ is $\family F$-safe, where $\structure G \oplus \mathbb H$ is the result of identifying $e_j$ and $f_j$ for all $1\leq j\leq k$.
        \item Let $(\structure G,\xi^\structure G, e_1),(\structure G,\xi^\structure G, e_2),(\structure H,\xi^\structure H,f)$ be $\family F$-safe.
        Suppose that there exist an $\family F$-free extension $\chi^{\structure G}$ of $\xi^{\structure G}$ and an $\family F$-free extension $\chi^{\structure H}$ of $\xi^{\structure H}$ such that $\chi^{\structure G}(e_1)=\chi^{\structure H}(f)$.
        Then $(\structure G \oplus \mathbb H,\xi^\structure G \oplus \xi^\structure H, e_2)$ is $\family F$-safe, where $\structure G \oplus \mathbb H$ is the result of identifying $e_1$ and $f$.
        \item\label{itm:cyclesafe}
        For every $i\in [c]$ and $e,e'\in E^{\mathbb F_i}$, the partially colored graph $(\structure F_i,(i\setminus ee'), e)$ is $\mathcal F$-safe.
        \item\label{item:4easyprop} For every color $i\in [c]$, edges $e,e' \in E^{\structure F_i}$ of $\mathbb F_i$, and a colored graph with a specified edge $(\structure H, \chi^{\structure H}, f^{\structure H})$ such that $\chi^{\structure H}(f^{\structure H}) \neq i$, we have that $(\structure F_i\oplus \structure H, (i\setminus ee') \oplus \chi^{\structure H}, e')$ is $\family F$-safe, where $\structure F_i\oplus \structure H$ is obtained by identifying $e$ with $f$.
    \end{enumerate}
\end{proposition}
\begin{proof}
    The first two items are a simple application of the associativity of the operation of gluing along an edge on colored graphs.
    Moreover, Item \ref{itm:cyclesafe} is a special case of Item \ref{item:4easyprop}, so we only prove the latter.
    For the sake of readability, define $\structure T := \structure F_i \oplus \structure H$, $\xi^{\structure T} := (i\setminus ee')\oplus \chi^{\structure H}$ and $f^{\structure T}:= e'$.
    Thus, we need to show that $(\structure T, \xi^{\structure T}, f^{\structure T})$ is $\family F$-safe.
    Let $(\structure G,\chi^\structure G )$ be an $\family F$-free colored graph with a distinguished edge $f^{\structure G}$ and such that there exists a full extension $\gamma$ of $\xi^{\structure{T}}$ with $\chi^\structure G(f^{\structure G})=\gamma(f^{\structure{T}})$.
    Let $\structure G \oplus \structure{T}$ be the amalgam of $\structure{T}$ and $\structure G$ over the edges $f^{\structure G}$ and $f^{\structure{T}}$, respectively.
    We show that the image of a hypothetical homomorphism from a member of $\family F$ to $(\structure G \oplus \structure{T},\chi^\structure G \oplus \gamma)$, has to be fully contained in $(\structure G,\chi^\structure G )$ or fully in $(\structure T, \gamma)$.
    Since $(\structure G,\chi^\structure G )$ and $(\structure T, \gamma)$ are $\family F$-free, this shows that $(\structure G \oplus \structure{T},\chi^\structure G \oplus \gamma)$ is $\family F$-free.
    Thus, take $(\structure F,\chi^{\structure F})\in\family F$.
    If $\structure F$ is a clique, there is nothing to prove.
    Otherwise, we are in the case that $\structure F$ is an odd cycle.
    By our second standing assumption discussed at the beginning of~\Cref{section:main_result}, we consider this hypothetical homomorphism $h$ to be injective.
    Suppose that the image of $h$ intersects both $\structure{T}$ and $\structure G$.
    We start by observing that the image must contain both vertices in the intersection otherwise we get a cut vertex in the image which combined with the fact that $\structure F$ is an odd cycle implies a contradiction.
    By the fact that the image must contain also vertices that are in $\structure{T}$ and not in $\structure G$ we get that it must contain an edge colored in $i$. 
    Moreover, the length of $\structure F$ has to be strictly larger than the length of $\structure F_i$, but this is not possible for the choice of $\structure F_i$ thus we obtain a contradiction.
\end{proof}

\section{Proof of dichotomy}\label{section:main_result}
We fix a family $\family F$ that respects the premise of~\Cref{thm:main}. We start by arguing that we can make two standing assumptions for the rest of the section.
First, without loss of generality, we can assume that for every color $i\in [c]$ there is an $\structure S$-free monochromatic obstruction $(\mathbb F_i,i)\in\mathcal F$; otherwise every $\structure S$-free graph has an $\family F$-free coloring obtained by coloring every edge with the color $i$.
Second, whenever $(\mathbb F_i, i)\in \mathcal F$ and the graph $\mathbb F_i$ is isomorphic to an odd cycle $\mathbb C_{\ell_i}$, we assume that $\mathcal F$ contains the $i$-monochromatic odd cycle $(\mathbb C_k, i)$, for every odd $3\leq k<\ell_i$; these obstructions are implied by $(\mathbb C_{\ell_i},i)$, so this does not change the problem $\Col(\mathcal F)$.
Thus, we can assume without loss of generality that for every edge-colored graph $(\mathbb G, \chi^\mathbb G)$, if $(\mathbb F_i, i)\to (\mathbb G, \chi^\mathbb G)$ for some monochromatic odd cycle $(\mathbb F_i, i)$ in $\mathcal F$, then there is an injective homomorphism from $(\mathbb F_i',i)$ to $(\mathbb G, \chi^\mathbb G)$ for some $(\mathbb F_i',i)\in\mathcal F$.

We prove \Cref{thm:main} by showing that $\Col(\family{F})$ is poly-time equivalent to an NP-complete finite-domain constraint satisfaction problem.
The proof consists of the following steps.
\begin{description}
    \item[Infinite-domain CSP] We first set up the tools that we need from the theory of infinite-domain constraint satisfaction to produce a reduction from $\Ext(\family F)$ to $\Col(\family{F})$.
    In graph-theoretic terms, the main result of this section (\Cref{cor:mcc_pp}) essentially states that for every $\mathcal F$-free edge-colored graph $(\mathbb G, \gamma^\mathbb G)$, there exists a graph $\mathbb D$ containing $\structure G$ as a subgraph and such that for every graph $\mathbb D'$ containing $\mathbb D$, if  $\mathbb D'$ has an $\mathcal F$-free coloring, then it has an $\mathcal F$-free coloring that extends $\gamma^\mathbb G$.
    This graph $\mathbb D$ is then used to prove that $\Ext(\family F)$ reduces in poly-time to $\Col(\family F)$.
    The existence of $\mathbb D$ is not constructive and is proved by compactness; in particular, the existence of $\mathbb D$ alone is not sufficient for a poly-time reduction.
    \item[Reduction from $\Ext(\family F)$ to $\Col(\family{F})$]
    We prove the existence of \emph{color-determiners}, which are partially colored graphs that impose that a certain fixed edge receives a specified color in any valid extension.
    Using those color-determiners, we show that it is enough to use a graph $\structure D$ as above in the case where $\structure G$ has constant size.
    \item[Reduction from finite-domain CSP to $\Ext(\family F)$]
    Here, we prove the existence of \emph{color-equality gadgets}, which are partially colored graphs that impose that two certain fixed disjoint edges receive the same color in all valid colorings.
    \item[Hardness proof] We prove that the finite-domain CSP equivalent to our problem is NP-complete using the algebraic approach to constraint satisfaction.
\end{description}
\subsection{Infinite-domain CSP}\label{sec:infinite_csp}
The goal of this section is to prove \Cref{cor:mcc_pp}.
First, let us introduce necessary model-theoretic notions which we use.

For $\tau$-structures $\mathbb A$ and $\mathbb B$, a homomorphism $e\colon \mathbb A \to \mathbb B$ is an \emph{embedding} if it is injective and, for every $R\in \tau$ of arity $k$ and every $\bar a\in A^k$, we have $\bar a \in R^\mathbb A$ if and only if $e(\bar a) \in R^\mathbb B$.
A surjective embedding is called an \emph{isomorphism}.
An isomorphism from a structure $\mathbb A$ to itself is called an \emph{automorphism}.
The set of all automorphisms of $\mathbb A$, denoted $\mathrm{Aut}(\mathbb A)$, forms a group which acts with permutations on $A^n$, for every $n>0$.
A $\tau$-structure $\mathbb A$ is \emph{$\omega$-categorical} if, for every $n>0$, the action $\mathrm{Aut}(\mathbb A)\curvearrowright A^n$ has finitely many orbits.

Let $\mathbb A_1$ and $\mathbb A_2$ be two edge-colored graphs such that they have a common induced (possibly empty) subgraph $\mathbb B$, i.e., for every $i\in\{1,2\}$ there exists an embedding $e_i\colon \mathbb B \hookrightarrow \mathbb A_i$.
An edge-colored graph $\mathbb C$ is an \emph{amalgam} of $\mathbb A_1$ and $\mathbb A_2$ over $\mathbb B$ if, for every $i\in\{1,2\}$ there exists an embedding $f_i\colon \mathbb A_i \hookrightarrow \mathbb C$ such that $f_1\circ e_1 = f_2 \circ e_2$.
A class $\mathcal C$ has the \emph{amalgamation property} (AP) if for every $\mathbb A_1, \mathbb A_2, \mathbb B\in \mathcal C$ and every $e_i\colon \mathbb B \hookrightarrow \mathbb A_i$, the class $\mathcal C$ contains an amalgam of $\mathbb A_1$ and $\mathbb A_2$ over $\mathbb B$.
A relational structure $\mathbb H$ is called \emph{homogeneous} if every isomorphism between its finite substructures extends to an automorphism of the whole structure.
An \emph{orbit} under $\Aut(\structure H)$ is a set of tuples of the form $(\alpha(a_1),\dots,\alpha(a_n))$ for some $(a_1,\dots,a_n)\in H^n$ and where $\alpha$ ranges over the elements of $\Aut(\structure H)$.
The homogeneity of $\structure H$ means exactly that every finite substructure of $\mathbb H$ defines a unique orbit of $\Aut(\mathbb H)$; thus, one can label every orbit of $\Aut(\mathbb H)$ by a finite structure.
The amalgamation property and homogeneity are connected by the following theorem of Fra\"{i}ss\'{e}.
\begin{theorem}[Th.~6.1.2 in~\cite{hodges_book}]\label{th:fraisse}
For every family $\mathcal C$ which is closed under taking substructures and which has the AP, there is a homogeneous structure $\mathbb H$ such that, for every finite structure $\mathbb G$, there is an embedding $\mathbb G \hookrightarrow \mathbb H$ if and only if $\mathbb G \in \mathcal C$.
The structure $\mathbb H$ is unique up to isomorphism and is called the \emph{Fra\"{i}ss\'e-limit} of $\mathcal C$.
\end{theorem}
Denote by $\mathrm{fm}(\mathcal F)$ the set of all Yes instances of $\Col(\mathcal F)$, i.e., all those graphs $\mathbb G$ such that at least one edge-coloring of $\mathbb G$ belongs to $\mathrm{Forb}(\mathcal F)$.
For relational signatures $\tau \subset \sigma$, say that a $\sigma$-structure $\mathbb B$ is a \emph{$\sigma$-expansion} of a $\tau$-structure $\mathbb A$, if $A=B$ and $R^\mathbb A = R^\mathbb B$ for every $R\in \tau$; in this case $\mathbb A$ is the \emph{$\tau$-reduct}.
A class of $\sigma$-structures $\mathcal D$ is an \emph{expansion} of a class of $\tau$-structures $\mathcal C$ if for every $\mathbb D\in\mathcal D$, the $\tau$-reduct $\mathbb D^\tau$ belongs to $\mathcal C$.
If $\mathcal D$ is an expansion of $\mathcal C$, then $\mathcal C$ is the \emph{reduct} of $\mathcal D$.

The first step is to construct an infinite graph $\mathbb G_\mathcal F$ such that $\mathrm{CSP}(\mathbb G_\mathcal F)$ and $\Col(\mathcal F)$ define the same problem.
A standard way of constructing $\mathbb G_\mathcal F$ is to take the Fra\"{i}ss\'{e}-limit (cf. \Cref{th:fraisse}) of a class which has the amalgamation property (AP).
In the case where all the graphs of $\mathcal F$ are cliques, the class of all finite $\mathcal F$-free edge-colored graphs $\mathrm{Forb}(\mathcal F)$ already has the AP.
In this case, the graph-reduct of the Fra\"{i}ss\'{e}-limit of $\mathrm{Forb}(\mathcal F)$ is the desired graph $\mathbb G_\mathcal F$.
However, if $\mathcal F$ has odd cycles, $\mathrm{Forb}(\mathcal F)$ does not have the AP.
For example, if two paths of the same color have lengths $2$ and $3$, they are both $\mathcal F$-free, but amalgamation over their endpoints will contain a monochromatic $\mathbb C_5$.
However, it is still possible to expand $\operatorname{Forb}(\mathcal F)$ with finitely many new relations so that the resulting expansion has the AP.
We describe this construction below.

A \emph{piece} of an edge-colored cycle $(\mathbb C_\ell, \chi^{\mathbb C_\ell})$ is a triple $(\mathbb P_k, \chi^{\mathbb P_k}, \bar r)$, where $\bar r = (r_1,r_2)\in (C_\ell)^2$ is a pair of distinct nonadjacent vertices of $\mathbb C_\ell$,   $\mathbb P_k$ is a path between $r_1$ and $r_2$ of length $k\in\{2,\ldots,\ell-2\}$, and $\chi^{\mathbb P_k}$ is the restriction of $\chi^{\mathbb C_\ell}$ on $E^{\mathbb P_k}$.
Let $\sigma$ be the signature consisting of the edge-relation $E$ and all the edge-colors $[c]$ seen as binary relations.
Let $\rho$ consist of binary relation symbols which are in one-to-one correspondence with the pieces of edge-colored cycles in $\mathcal F$.
Let $\mathcal{HN}$ be the set of all (induced) substructures of every finite $(\sigma\cup\rho)$-structure $\mathbb G$ that satisfy the following properties:
\begin{itemize}
    \item the $\sigma$-reduct $\mathbb G^\sigma$ is in $\mathrm{Forb}(\mathcal F)$,
    \item for every piece $(\mathbb P, \chi^{\mathbb P}, \bar r)$, there is a homomorphism $h$ from $(\mathbb P, \chi^{\mathbb P})$ to $\mathbb G^\sigma$ if and only if $(h(r_1),h(r_2))$ is contained in the corresponding $\rho$-relation of $\mathbb G$.
\end{itemize}
Expanding the relational signature with  $\rho$-relations in order to get the AP is known as the \emph{Hubi\v{c}ka-Ne\v{s}et\v{r}il-construction} due to their results in~\cite{hubicka_nesetril2015,hubicka_nesetril2016}.
It is clear that the class $\mathcal{HN}$ is closed under taking substructures.
The proof that $\mathcal{HN}$ has the AP is the same as, for example, the proof of Theorem 4.3.8 in~\cite{Bodirsky_book}.
Let $\mathbb{HN}$ be the Fra\"{i}ss\'e-limit of $\mathcal{HN}$.

A countably infinite structure $\mathbb A$ is a \emph{model-complete core} if, for every endomorphism $e\in \mathrm{End}(\mathbb A)$ and every finite $F\subset A$, there is an automorphism $\alpha \in \Aut(\mathbb A)$ such that $e\restriction_F = \alpha\restriction_F$.
Let $(\mathbb C,\not=)$ be the model-complete core of $(\mathbb{HN},\not=)$, where $\not=$ stands for the inequality relation.
Let $\mathcal C$ be the set of finite substructures of $\mathbb C$.
Note that $\structure C$ is homogeneous, and therefore the orbits of pairs of distinct elements under $\Aut(\structure C)$ correspond to the 2-element substructures of $\mathbb C$. We give below a combinatorial description of such orbits.

For a $(\sigma\cup\rho)$-structure $\mathbb A$, denote by $\rho^{-1}(\mathbb A)$ the $\sigma$-structure obtained from $\mathbb A$ by replacing each $\rho$-edge $xy$ with a copy of the corresponding piece with endpoints in $x$ and $y$.
For $X\subseteq \rho$ and a color $i\in [c]$, let $\mathbb E_X^i$ be a copy of $\mathbb K_2$ colored with $i$ and where every relation in $X$ is imposed on the vertices.
We call such structures \emph{$(\sigma\cup\rho)$-edges}.
Call a $(\sigma\cup\rho)$-edge $\mathbb E_X^i$ \emph{maximal} if $\rho^{-1}(\mathbb E_X^i)$ is $\mathcal F$-free and $\rho^{-1}(\mathbb E^i_Y)$ is not $\mathcal F$-free, for every $Y\supsetneq X$.
\begin{proposition}\label{prop:mccmaximal}
Every $(\sigma\cup\rho)$-edge of $\mathbb C$ is maximal.
\end{proposition}
\begin{proof}
The proof is essentially the same as the proof of Lemma 4.23 in~\cite{MMSNPjournal}.
Indeed, suppose that $\mathbb C$ contains a $(\sigma\cup\rho)$-edge $\structure E^i_X$ that is not maximal.
The fact that it is not maximal simply means that for some $Y\subseteq\rho$ we have $X\subsetneq Y$ and that $\mathbb E_Y^i$ is also a substructure of $\mathbb C$.
Let $\bar s\in C^2$ be in the 2-orbit defined by some $(\sigma\cup\rho)$-edge $\mathbb E_X^i$ and $(\mathbb P,\chi^\mathbb P, \bar r)$ be a piece corresponding to a relation in $Y\setminus X$, then let $\mathbb C'$ be obtained from $\mathbb C$ by attaching a copy of $(\mathbb P,\chi^\mathbb P, \bar r)$ to $\bar s$ and adding $\bar s$ to the corresponding $\rho$-relation $S$.
Note that $(\mathbb C,\neq)$ maps homomorphically to $(\mathbb C',\neq)$ by the identity mapping.
Conversely, every finite substructure of $(\structure C',\neq)$ has a homomorphism to $(\mathbb C,\neq)$, and therefore by compactness we obtain a homomorphism $(\structure C',\neq)\to(\structure C,\neq)$.
Composing these two homomorphisms gives us an endomorphism of $(\mathbb C,\neq)$ which must be an embedding since $(\mathbb C,\neq)$ is a core.
Thus, the identity map is an embedding $(\mathbb C,\neq)\to(\mathbb C',\neq)$
but the image of $\bar s$ belongs to the relation $S$ while $\bar s$ does not, which is a contradiction.
\end{proof}

For two structures $\mathbb A, \mathbb B \in \mathcal C$, denote by $\binom{\mathbb B}{\mathbb A}$ the set of all embeddings $\mathbb A \hookrightarrow \mathbb B$.
A class $\mathcal C$ of finite structures  has the \emph{Ramsey property} (RP) if for every $\mathbb A,\mathbb B \in \mathcal C$ and every $n>0$, there is $\mathbb C \in \mathcal C$ such that for every $\chi \colon \binom{\mathbb C}{\mathbb A} \to [n]$, there is $f\in \binom{\mathbb C}{\mathbb B}$ such that for each two $g_1,g_2\in \binom{\mathbb B}{\mathbb A}$, we have $\chi(f\circ g_1) = \chi(f\circ g_2)$.
A relational structure is \emph{Ramsey} if the class of its finite (induced) substructures has the Ramsey property.

Let $\mathcal C^<$ be the class of all $(\sigma\cup\rho\cup\{<\})$-structures whose $(\sigma\cup\rho)$-reduct is in $\mathcal C$ and whose $<$-reduct is isomorphic to a linear order.
Let $(\mathbb C, <)$ expand $\mathbb C$ with a generic linear order, i.e., $(\mathbb C, <)$ is the Fra\"iss\'e-limit of $\mathcal C^<$.
By Theorem 4 in~\cite{cores_ramsey}, the structure $(\mathbb C, <)$ is Ramsey.

Let $\mathcal O$ be the set of all linearly ordered $(\sigma\cup\rho)$-edges $(\mathbb E_X^i,<)$ in $\mathcal C^<$.
Call $r\colon \mathcal O \to \mathcal O$ a \emph{recoloring} if there is  $\hat r \colon \mathcal C^< \to \mathcal C^<$ such that, for every $\mathbb B\in \mathcal C^<$, $\hat r(\mathbb B)$ is a structure on the same set of vertices $B$, and for every $(\mathbb E_X^i,<)\in \mathcal O$ and every $e\colon E_X^i \to B$, $e$ is an embedding $(\mathbb E_X^i,<)\hookrightarrow\mathbb B$ if and only if $e$ is an embedding of $r(\mathbb E_X^i,<)$ into $\hat r(\mathbb B)$.
In order to prove \Cref{cor:mcc_pp}, we need first to prove \Cref{lem:recoloring_bijective} which assures that every recoloring, for given $\mathcal F$, must be bijective.
Proving \Cref{lem:recoloring_bijective} is the technical core of \Cref{sec:infinite_csp}.

If $\mathbb S = \mathbb C_\ell$, then put $m:=\ell$, and if $\mathbb S = \mathbb K_\ell$, then put $m:=0$.
Note that all $\rho$-relations of maximal $(\sigma\cup\rho)$-edges are symmetric, which is obvious for monochromatic pieces.
If $\structure S$ is an odd cycle, being $\mathbb S$-free simply means that every edge-coloring of $\mathbb S$ is forbidden, and thus if $\rho^{-1}(\mathbb E^i_X)$ has a non-monochromatic path $\mathbb P_k$ attached to the edge $E^i_X$, then we also have $(\mathbb P_k,\chi)$ attached to the edge $E^i_X$, for every possible edge-coloring $\chi\colon E^{\mathbb P_k}\to [c]$ of $\mathbb P_k$.
This implies that $(\mathbb E_X^i,<)$ is isomorphic to $(\mathbb E_X^i,>)$, so $\hat r$ defines a mapping from $\mathcal C$ to $\mathcal C$ and we can omit the sign $<$ in the context of recoloring.

Now, for $i\in [c]$, let $\mathbb E^i_\mathrm{half}$ be the $(\sigma\cup\rho)$-edge whose $\sigma$-reduct is the $i$-colored edge and which induces the following $\rho$-relational tuples on its two vertices.
\begin{itemize}
    \item For every $j\not=i$ such that $\mathbb F_j=\structure C_{\ell_j}$ is an odd cycle and every $k$ such that $\max(m,\ell_j/2+1)\leq k\leq \ell_j-2$, $\mathbb E^i_\mathrm{half}$ has the $(\mathbb P_{k}, j)$-tuple.
    \item For every odd $k\leq m-2$ and every $\chi\colon E^{\mathbb P_k}\to [c]$, $\mathbb E^i_\mathrm{half}$ has the $(\mathbb P_k, \chi)$-tuple.
    \item If $\mathbb F_i=\mathbb C_{\ell_i}$ is an odd cycle, then, for every odd $3\leq k \leq \ell_i-2$, $\mathbb E^i_\mathrm{half}$ has the $(\mathbb P_k, i)$-tuple.
\end{itemize}

\begin{lemma}\label{claim:e_half_is_maximal}
There is a unique maximal $(\sigma\cup\rho)$-edge with this property.
\end{lemma}
\begin{proof}
    Suppose that, for some $i\in[c]$, $\mathbb E^i_\mathrm{half}$ is not maximal.
    This implies that, for some $j\not=i$ and for some $k' < \max(m,\ell_j/2+1)$ we can attach the path $(\mathbb P_{k'}, \gamma)$ to $\rho^{-1}(\mathbb E^i_\mathrm{half})$ and the result will be $\mathcal F$-free.
    If $k'<m$ and $k'$ is odd and $\gamma$ is non-monochromatic, then this $\rho$-relation is already in $\mathbb E^i_\mathrm{half}$.
    If $k'< m$ and $k'$ is even, then the graph contains an odd cycle of length at most $m$ and thus is not $\mathcal F$-free.
    If $\gamma$ is $j$-monochromatic, for some $j\in[c]$, and if $k' \leq \ell_j/2$, then $\ell_j-k' > \ell_j/2$, so $\rho^{-1}(\mathbb E^i_\mathrm{half})$ already contains a copy of $(\mathbb P_{\ell_j-k'},j)$, and gluing it with $(\mathbb P_{k'}, \gamma)$ results in $(\mathbb C_{\ell_j}, j)$, so the result is not $\mathcal F$-free.
\end{proof}

\begin{lemma}\label{proposition:Ehalf}
    Let $\structure E^i_{\mathrm{full}}$ consist of two copies of $\structure E^i_{\mathrm{half}}$ glued on a vertex, and let $u,v$ be the two vertices of degree $1$. 
    Let $j$ be a color other than $i$.
    Then every path of $j$-colored edges from $u$ to $v$ in $\rho^{-1}(\structure E^i_{\mathrm{full}})$ has length greater than the length of all cycles of color $j$ in $\family F$.
\end{lemma}
\begin{proof}
    Let $(\structure C_\ell,j)$ be a cycle in $\family F$ of color $j\neq i$.
    The shortest $j$-monochromatic path from $u$ to $v$ has length at least $2(\ell_j/2 + 1)>\ell_j$.
\end{proof}

\begin{lemma}\label{lem:retraction}
    Let $r$ be a recoloring.
    Then, for some $n>0$, $r^n$ is a \emph{retraction}, i.e., for every $\mathbb E^i_X\in r(\mathcal O)$, we have $r^n(\mathbb E^i_X) = \mathbb E^i_X$.
\end{lemma}
\begin{proof}
    This simply follows from the fact that any map $r\colon A\to A$ where $A$ is a finite set is such that for some $m>0$, $r^m$ induces a permutation on its image, and then for $m'$ large enough $(r^m)^{m'}=r^{mm'}$ is the identity function on its image.
\end{proof}

Further, we assume without losing any generality that $r$ is already a retraction.

\begin{proposition}\label{prop:base_half}
    Let $r$ be a recoloring such that $r(\mathbb E^i_\mathrm{half}) = \mathbb E^j_X$.
    Then, either $i=j$ or ($\mathbb F_i \to \mathbb F_j$ and $\mathbb F_j \not\to \mathbb F_i$).
\end{proposition}
\begin{proof}
    For the contrary, assume that $r(\mathbb E^i_\mathrm{half}) = \mathbb E^j_X$ for $i\neq j$ and that $\mathbb F_i \not\to \mathbb F_j$. 
    Notice that the latter implies in particular that $\structure F_j \to \structure F_i$ since homomorphism preorder forms a linear order once restricted on odd cycles and cliques.
    
    \textbf{Case $\mathbb F_j=\mathbb C_{\ell_j}$ is an odd cycle.}
    Define $\structure W$ in the signature $\sigma\cup \rho$ as the cycle $\mathbb C_{\ell_j}$ where we assign $\mathbb E^i_\mathrm{half}$ on two adjacent edges and $\mathbb E^j_X$ on the remaining ones.
    We prove that $\rho^{-1}(\structure W)$ is $\family F$-free and this would contradict the fact that $r$ is a recoloring since $\rho^{-1}(\hat r(\structure W))$ contains $(\mathbb C_{\ell_j}, j)$ and thus is not $\family F$-free.
    To prove that $\rho^{-1}(\structure W)$ is $\family F$-free, we assume that an obstruction maps to it and we show that this cannot happen. 
    Since $\ell_j > m$ and since both $\mathbb E^i_\mathrm{half}$ and $\mathbb E^j_X$ are in $\mathcal C$, a non-monochromatic obstruction cannot map to $\rho^{-1}(\structure W)$ and if some monochromatic obstruction maps to $\rho^{-1}(\structure W)$, then all vertices of $\mathbb C_{\ell_j}$ must be in the image.
    For any color $k\not=i$, the image of a $k$-monochromatic cycle must go through some paths attached to both $\mathbb E^i_\mathrm{half}$-edges but this is not possible by \Cref{proposition:Ehalf}.
    So, we can assume that the obstruction is $i$-monochromatic. 
    If $\structure F_i$ is a clique, then there is nothing to prove. 
    If $\structure F_i = \mathbb C_{\ell_i}$ is an odd cycle, then, since the $\sigma$-reduct of $\mathbb E^j_X$ is a $j$-colored edge and $j\not=i$, the homomorphic image of $\mathbb C_{\ell_i}$ must go through some paths attached to  $\mathbb E^j_X$-edges, which cannot happen because we assumed $\structure F_j \to \structure F_i$.

    \textbf{Case $\mathbb F_i = \mathbb K_{\ell_i},\;\mathbb F_j=\mathbb K_{\ell_j}$ -- cliques.}
    Take a copy of $\mathbb K_{\ell_j}$, choose a vertex $v\in K_{\ell_j}$ and assign the edges incident to $v$ to the type $\mathbb E^i_\mathrm{half}$ and assign the remaining edges to the type $\mathbb E^j_X$, let $\mathbb W$ be the result.
    If $\rho^{-1}(\mathbb W)$ is not $\mathcal F$-free, then it must contain a forbidden $j'$-monochromatic odd cycle for some $j'\in[c]$.
    That is, for some vertices $w_1,\ldots,w_k\in W$, the edges $w_1w_2, w_2w_3, \ldots, w_kw_1$ have some $j'$-monochromatic paths attached to them.
    If some $w_s=v$, then both edges $w_{s-1}w_s$ and $w_sw_{s+1}$ have type $\mathbb E^i_\mathrm{half}$, and by applying \Cref{proposition:Ehalf} the length $k$ will be strictly greater than $\ell_{j'}$, so such a cycle cannot be forbidden.
    As no $w_s=v$, every $\mathbb E^j_X$-edge has a $j'$-monochromatic path of odd length smaller than $\ell_{j'}/3$ attached to it, so it suffices to take any three edges $w_1w_2, w_2w_3, w_3w_1$ and they induce a $j'$-monochromatic cycle whose length is smaller than $\ell_{j'}$.
    Now take a copy of $\mathbb K_3$ and assign two edges to the type $\mathbb E^i_\mathrm{half}$ and the remaining one to $\mathbb E^j_X$.
    The $\rho^{-1}$-image of this is $\mathcal{F}$-free because every odd cycle that maps into it has to intersect both copies of $\mathbb E^i_\mathrm{half}$ in its image and we know that this is not possible, see \Cref{proposition:Ehalf}. 
    But $\hat r$ sends it to $\mathbb K_3$ made of  $\mathbb E^j_X$, which is not $\family{F}$-free hence contradicts $r$ being a recoloring.
\end{proof}

\begin{lemma}\label{lemma:evenhalf}
    For every $i\in[c]$ there is $X\subseteq\rho$ such that $r(\mathbb E^i_\mathrm{half}) = \mathbb E^i_X$.
\end{lemma}
\begin{proof}
    The proof goes by induction on the reverse linear order $\leq$ defined on the set of monochromatic obstructions $\{\structure F_i \mid i \in [c]\}$, as $\structure F_j\leq \structure F_i \Leftrightarrow \structure F_i \to \structure F_j$. So it is the opposite of the homomorphism preorder. 
    The ``base'' case follows from \Cref{prop:base_half}.
    The inductive case is shown by contradiction, so suppose $r(\mathbb E^i_\mathrm{half}) = \mathbb E^j_X$, for $j\not=i$.
    
    \textbf{Case $\mathbb F_j = \mathbb C_{\ell_j}$ -- odd cycle.}
    By \Cref{prop:base_half}, we have $\structure F_i \to \structure F_j \land \structure F_j \not\to \structure F_i$ which implies that $\structure F_i = \structure C_{\ell_i}$ and $\ell_i > \ell_j$.
    We define $\mathbb W$ in the signature $\sigma \cup \rho$ as a copy of $\mathbb C_{\ell_j}$ on which we assign $\mathbb E^i_\mathrm{half}$ on two edges and $\mathbb E^j_\mathrm{half}$ on the remaining ones.
    Note that, $\rho^{-1}(\mathbb W)$ is $\family F$-free, since no monochromatic cycle of any color can map to it, thanks to \Cref{proposition:Ehalf}.
    However, the $\sigma$-reduct of $\rho^{-1}(\hat{r}(\structure W))$ is $(\mathbb C_{\ell_j},j)$ which is not $\family F$-free, thus we find the desired contradiction.
    
    \textbf{Case $\mathbb F_j = \mathbb K_{\ell_j}$ -- clique.}
    Take a copy of $\mathbb K_{\ell_j}$, choose a vertex $v\in K_{\ell_j}$ and assign $\mathbb E^i_\mathrm{half}$ to all the edges incident to $v$ and $\mathbb E^j_\mathrm{half}$ to the rest, let $\mathbb W$ be the result.
    Note that the $\sigma$-reduct of $\rho^{-1}(\hat{r}(\mathbb W))$ is $(\mathbb K_{\ell_j},j)$ which is not $\family F$-free. So if we show that $\rho^{-1}(\mathbb W)$ is $\family F$-free we would get that $r$ cannot be a recoloring.
    Note that the $\sigma$-reduct of $\mathbb W$ is non-monochromatic, therefore $\mathcal F$-free.
    So it only remains to consider monochromatic odd cycles as possible obstructions. But the image of a hypothetical monochromatic odd cycle would have to intersect either two structures $\mathbb E^i_\mathrm{half}$ or two structure $\mathbb E^j_\mathrm{half}$. 
    Anyhow, thanks to \Cref{proposition:Ehalf}, none of those can happen. So, $\rho^{-1}(\mathbb W)$ is $\family F$-free and the proof is concluded.
\end{proof}

\begin{lemma}\label{lem:recoloring_bijective}
Every recoloring from $\mathcal C^<$ to $\mathcal C^<$ is bijective.
\end{lemma}
\begin{proof}
    Let $\mathbb E^i_X$ be a maximal $(\sigma\cup\rho)$-edge and suppose first that $r(\mathbb E^i_X) = \mathbb E^j_Y$ for some $j\not=i$.
    
    \textbf{Case $\mathbb F_j = \mathbb C_{\ell_j}$ -- an odd cycle.}
    Take a copy of $\mathbb C_{\ell_j}$ and assign two of its edges to the type $\mathbb E^j_\mathrm{half}$ and the remaining edges to the type $\mathbb E^i_X$ and call the result $\structure W$.
    Note that $\rho^{-1}(\structure W)$ is $\mathcal F$-free because all monochromatic paths are too long (as there are two $\mathbb E^j_\mathrm{half}$-edges) but the $\sigma$-reduct of its $\hat r$-image is the $j$-monochromatic $\mathbb C_{\ell_j}$, by \Cref{lemma:evenhalf}.
   
    \textbf{Case $\mathbb F_j = \mathbb K_{\ell_j}$ -- a clique.}
    Take a copy of $\mathbb K_{\ell_j}$ and assign one of its edges to $\mathbb E^i_X$ and the rest to $\mathbb E^j_\mathrm{half}$ and call the result $\mathbb W$.
    Suppose that $\rho^{-1}(\structure W)$ is not $\mathcal F$-free, which means that some forbidden $j'$-monochromatic odd cycle $(\mathbb C_{\ell_{j'}}, j')$ maps homomorphically to it.
    Its image must contain at least three different edges of $\mathbb K_{\ell_j}$ and thus at least two of them must be of the type $\mathbb E^j_\mathrm{half}$, therefore the $p$-monochromatic cycle has length $>\ell_{j'}$, a contradiction.
    This implies that $\rho^{-1}(\structure W)$ is $\mathcal F$-free but the $\sigma$-reduct of $\hat r(\mathbb W)$ is $(\mathbb K_{\ell_j}, j)\in\mathcal F$.

    Finally, suppose that $r(\mathbb E^i_X) = \mathbb E^i_Y$, for some $Y\not=X$.
    The fact that $X\not=Y$ means, by maximality, that, for some $j\not=i$ and for some $k < \ell_j$, we can attach $(\mathbb P_k, j)$ to an $\mathbb E^i_X$-edge and that we can attach $(\mathbb P_{\ell_j-k}, j)$ to $\mathbb E^i_Y$.
    Take a copy of $\mathbb C_{k+1}$ and assign $\mathbb E^i_X$ to one of its edges and $\mathbb E^j_\mathrm{half}$ to the rest, call the result $\mathbb W$.
    Note that $\rho^{-1}(\structure W)$ is $\mathcal F$-free because all monochromatic paths are too long (by \Cref{proposition:Ehalf}, as there are at least two $\mathbb E^j_\mathrm{half}$-edges) and because we can attach $(\mathbb P_k,j)$ to an $\mathbb E^i_X$-edge.
    However, by \Cref{lemma:evenhalf}, $\rho^{-1}(\hat r(\mathbb W)$ contains a copy of $(\mathbb C_{\ell_j},j)$ obtained by joining $(\mathbb P_k, j)$ and $(\mathbb P_{\ell_j-k}, j)$ by their endpoints.
\end{proof}

Let $\mathscr G$ and $\mathscr H$ be two groups acting with permutations on sets $A$ and $B$, respectively.
Call a homomorphism $h\colon \mathbb A\to \mathbb B$ \emph{canonical with respect to $\mathscr G$ and $\mathscr H$} if for every $n>0$, $\bar a\in A^n$, and $\alpha \in \mathscr G$, there is $\beta\in\mathscr H$ such that $h(\bar a) = \beta(h(\alpha(\bar a)))$.
That is, a canonical mapping induces a well-defined mapping from the orbits of $\mathscr G\curvearrowright A^n$ to the orbits of $\mathscr H\curvearrowright B^n$, for every $n>0$.
The only reason why the Ramsey property is necessary in this paper is that it ensures the existence of a \emph{canonical} homomorphism by the following lemma.

\begin{lemma}[cf.~\cite{bodirsky_pinsker_tsankov}, Th.\ 5 in~\cite{bodirsky_pinsker_ramsey_canonical}]\label{lem:canonization}
    Let $\mathbb A$ and $\mathbb B$ be countable $\omega$-categorical structures such that $\mathbb A$ is Ramsey, and let $\mathbb A',\mathbb B'$ be reducts of these structures.
    If there exists a homomorphism $h'\colon \mathbb A' \to \mathbb B'$, then there is also a homomorphism 
    \[
    h \in \overline{\{\beta \circ h' \circ \alpha \mid \alpha \in \Aut(\mathbb A),\; \beta\in\Aut(\mathbb B)\}}
    \]
    which is canonical with respect to $\Aut(\mathbb A)$ and $\Aut(\mathbb B)$.
\end{lemma}
Here, for a set of functions from $X$ to $Y$, we denote by $\overline{C}$ the set of all $g\colon X\to Y$ such that for every finite $S\subseteq X$ there is $f\in C$ satisfying $f\restriction_S = g\restriction_S$.

\begin{lemma}\label{lem:gfmcc}
    Let $\mathbb G_\mathcal F$ be the graph-reduct of $\mathbb C$.
    Then, the structure $(\mathbb G_\mathcal F, \not=)$ is a model-complete core.
\end{lemma}
\begin{proof}
Let $(\mathbb G,\neq)$ be the model-complete core of $(\mathbb G_\mathcal F, \not=)$, seen as a substructure of $(\mathbb G_{\mathcal F},\neq)$ and let $g\in\mathrm{End}(\mathbb G_\mathcal F, \not=)$ be a homomorphism $(\structure G_{\mathcal F},\neq)\to(\mathbb G,\neq)$.
As $(\mathbb G_\mathcal F, \not=)$ has a Ramsey expansion $(\mathbb C, <)$, we can use the \emph{canonisation lemma} (\Cref{lem:canonization}) and thus assume that $g$ is canonical with respect to $\Aut(\mathbb C, <)$.
By \Cref{lem:recoloring_bijective}, the action of $g$ on the orbits of edges of $(\mathbb C, <)$ is bijective, and thus $g^n\in\mathrm{End}(\mathbb C, <)$ for some $n>0$.
As $(\mathbb C, <)$ is a model-complete core, $g^n$ preserves all first-order formulas over $\mathbb C$.
Hence, $g$ and $g^{n-1}$ \emph{locally invert each other} in the sense of~\cite{Bodirsky_book}, and $g\in\overline{\Aut(\structure G_\family{F})}$ by Cor.~3.4.13 in~\cite{Bodirsky_book}, so it preserves all first-order formulas in $\structure G_\family{F}$.
This shows in particular that $(\mathbb G,\neq)$ and $(\mathbb G_\mathcal F, \not=)$ have the same first-order theory and are isomorphic by $\omega$-categoricity.
\end{proof}

Say that a set $S\subseteq (G_\mathcal F)^n$ is \emph{pp-definable} in $(\mathbb G_\mathcal F,\not=)$ if there is a finite graph $\mathbb D_S$ and a tuple $\bar d\in (D_S)^n$ such that, for every $\bar x\in (G_\mathcal F)^n$, there is an injective homomorphism $h\colon\mathbb D_S \to \mathbb G_\mathcal F$ such that $h(\bar d)=\bar x$ if and only if  $\bar x\in S$.

\begin{corollary}[Th. 4.5.1 in~\cite{Bodirsky_book}]\label{cor:mcc_pp}
    For every $n>0$ and every $\bar d \in (G_\mathcal F)^n$, the $\Aut(\mathbb G_\mathcal F, \not=)$-orbit of $\bar d$ is pp-definable in $(\mathbb G_\mathcal F, \not=)$.
\end{corollary}

\subsection{Reduction from Ext to Col}\label{section:coldet}

Let $\structure G$ be a graph, and $\xi^\mathbb G$ be a partial coloring of the edges of $\structure{G}$ in $c$ colors and $e\in E^\structure G \setminus\dom{\xi^{\structure{G}}}$ be an uncolored edge.
For $i\in [c]$, we call $(\structure{G}, \xi^{\structure{G}}, e)$ \emph{a color-$i$-determiner} if it is $\family{F}$-safe and
\begin{enumerate}[(1)]
    \item one of the two vertices incident to $e$ is not incident to any edge of $\dom{\xi^{\structure{G}}}$,
    \item there is an $\family{F}$-free extension of $\xi^{\structure{G}}$, and if $\gamma$ is an $\family{F}$-free extension of $\xi^{\structure{G}}$, then $\gamma(e)=i$.
\end{enumerate}
For $d\in\mathbb{N}$, a color-determiner $(\structure{G},\xi^{\structure{G}},e)$ is called \emph{$d$-remote} if $\dist(e, \dom{\xi^{\structure{G}}})\geq d$.
Now we start the technical part by proving the existence of a graph that has a valid coloring and a vertex for which the edges incident to it have a special pattern for all valid colorings.
\begin{lemma}\label{lemma:step1}
   There exists a graph $\structure H$ with the following properties:
\begin{enumerate}[(1)]
    \item $\structure H$ is $\family{F}$-free colorable;
    \item \label{item:step1} there is a vertex $x \in H$ such that if $\gamma$ is a $\family{F}$-free coloring of $\structure H$, then for each color $i$ there exists an edge $e$ incident to $x$ such that $\gamma(e)=i$.
\end{enumerate}
\end{lemma}
\begin{proof}
    We first show the existence of a graph $\structure G$ that is $\structure S$-free and not $\family F$-free colorable.
    Let $\structure M$ be the $\rightarrow$-maximum of the monochromatic part of $\family F$, which is necessarily a clique or a cycle of odd length $d\geq 5$. 
    Observe that for each color $i$, the colored graph $(\structure M, i^{\structure M})$ is in $\forb{\family F}$, where $i^{\structure M}$ is the constant coloring whose image is $i$.
    
    Suppose first that $\structure M$ is a clique $\structure K_m$.
    By Theorem 1 in~\cite{NESETRIL1976243}, there exists a $\structure K_{m+1}$-free graph $\structure G$ such that every edge-coloring of $\structure G$ contains a monochromatic $\structure K_{m}$ as a subgraph, and in particular $\structure G$ is not $\family F$-free colorable.
    Note that by the assumption on $\family F$ such a $\structure G$ has to be also $\structure S$-free, because $\structure K_{m+1}$ is an induced subgraph of $\structure S$.
    If $\structure M$ is a cycle $\structure C_d$ for some odd $d\geq 5$, then by Theorem 1.1 in \cite{HAN2015457} there exists a graph $\structure G$ with girth $d$
    and such that every edge-coloring of $\structure G$ contains a monochromatic copy of $\structure C_d$.
    Note that $\structure G$ is $\structure S$-free, since the girth of $\structure S$ is at most $d-2$.
    
    Removing edges one at a time, we reach a graph $\structure G'$ that is minimal with the property of not being $\family F$-free colorable; i.e., for every edge $e \in E^{\structure G'}$, we have that $\structure G'\setminus e$ is $\family{F}$-free colorable.
    Let $e$ be an arbitrary edge of $\structure G'$.
    We show that $\structure H:=\structure G'\setminus e$ satisfies the property of~\Cref{lemma:step1}.
    
    The first property is immediate. 
    We prove the second property by contradiction.
    Let $x$ be one of the two vertices of $e$. 
    Let $\gamma^\structure H$ be a $\family{F}$-free coloring of $\structure H$. 
    Suppose that there exists $i \in [c]$ such that every edge $f$ incident to $x$ satisfies $\gamma^\structure H (f) \neq i$.
    Define a coloring $\gamma^{\structure G'}$ of the edges of $\structure G'$ by extending it by $\gamma^{\structure G'}(e)=i$,
    so that $(\structure H,\gamma^\structure H)$ is a colored subgraph of $(\structure G',\gamma^{\structure G'})$.
    Let us assume by contradiction that there exist $(\structure F,\chi^{\structure F}) \in \family F$ and an injective homomorphism $h\colon (\structure F,\chi^{\structure F}) \rightarrow (\structure G',\gamma^{\structure G'})$ (recall our standing assumption made at the beginning of~\Cref{section:main_result}, allowing us to assume that this homomorphism is injective).
    Moreover, there exists $(u,v)\in E^\structure F$ such that $(h(u),h(v))=e$,
    otherwise $(\structure H,\gamma^\structure H)$ is not $\family F$-free. Say $h(v)=x$.
    Since $\gamma^{\structure G'}(e)=i$, we obtain that $(\structure F,\chi^{\structure F})$ is monochromatic of color $i$.
    Let $w\neq u$ be a vertex adjacent to $v$ in $\structure F$, we know that it exists because $\structure F$ is an odd cycle.
    Since every $f\neq e$ incident to $x$ satisfies $\gamma^{\structure G'}(f)\neq i$, we must have $h(w)=h(u)$, which contradicts the fact that $h$ is injective.
\end{proof}

From now until Lemma~\ref{lemma:step2c} we fix a color $i\in [c]$.
By \Cref{lemma:step1} (item~\ref{item:step1}), for every $\mathcal F$-free coloring there is some edge incident to $x$ which is colored with $i$.
The missing step to obtain a color-$i$-determiner is to force a \emph{fixed} edge to be colored with $i$ in all $\family F$-free colorings. We achieve this in several steps in the following lemmas.

\begin{lemma}\label{lemma:step2a}
    There exist a partially colored graph $(\structure I ,\xi^{\structure I})$ and $f^{\structure I} \in E^{\structure I} \setminus \dom{\xi^{\structure I}}$ such that:
    \begin{enumerate}[(1)]
        \item There exists an $\family F$-free coloring $\gamma$ of $\structure I$ that extends $\xi^{\structure I}$.
        \item For every $\family F$-free extension $\gamma$ of $\xi^{\structure I}$, we have $\gamma(f^{\structure I})=i$.
        \item There exists $x\in I$ incident to all edges in $\dom{\xi^{\structure I}}$ and $f^{\structure I}$.
    \end{enumerate}
\end{lemma}
\begin{proof}
    Take $\structure I$ to be the graph given by~\Cref{lemma:step1} and iteratively color the edges incident to $x$ by colors different from $i$ while there exists an $\family F$-free extension of the coloring.
    The graph $\structure I$ is depicted in~\Cref{fig:ci}.
    By item 2 in~\Cref{lemma:step1}, it is not possible to color all the edges incident to $x$ this way and therefore at least one edge $f^{\structure I}$ remains uncolored.
    By construction, any coloring extending this partial coloring of $\structure I$ must color $f^{\structure I}$ with the color $i$, and moreover only edges incident to $x$ are in the domain of $\xi^{\structure I}$.
\end{proof}

\begin{figure}[h]
        \centering
        \includegraphics[scale=1]{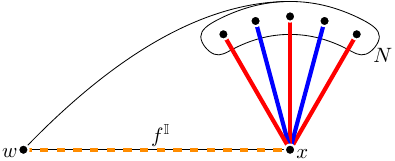}
        \caption{The graph $\structure I$ as in the proof of~\Cref{lemma:step2a}. The bottom edge must be colored with $i$ (taken here to be orange) in every $\family F$-free extension of this partial edge-coloring of $\structure I$.}
        \label{fig:ci}
\end{figure}

If $\mathcal F$ consisted of edge-colored cliques, then $(\structure I ,\xi^{\structure I}, f^{\structure I})$ would already be a color-$i$-determiner.
However, if $\mathcal F$ contains monochromatic cycles, then it is possible that $(\structure I ,\xi^{\structure I}, f^{\structure I})$ is not $\mathcal F$-safe. 
We ensure this property in the following steps and show that color-$i$-determiners exist for every $\mathcal F$ in our scope.
\begin{lemma}\label{propositio:Dnegi}
    There exist a partially colored graph $(\structure{T} ,\xi^{\structure{T}})$ and  $f^{\structure{T}}\in E^\structure{T}\setminus\dom{\xi^{\structure{T}}}$ such that:
    \begin{enumerate}[(1)]
        \item\label{item:Dnegi3} For every $\family F$-free extension $\gamma$ of $\xi^{\structure{T}}$, we have $\gamma(f^{\structure{T}})\neq i$. 
        \item\label{item:Dneg-soundness} For every $j\in [c]$ with $j\neq i$, there is an $\family F$-free coloring $\gamma$ of $\xi^{\structure T}$such that $\gamma (f^{\structure{T}}) = j$.
        \item \label{item:Dnegi4} $(\structure{T}, \xi^{\structure{T}}, f^{\structure{T}})$ is $\family F$-safe.
        \item\label{item:Dnegi-distance} One of the vertices of $f^{\structure{T}}$ is not incident to any edge in $\dom{\xi^{\structure{T}}}$.
    \end{enumerate}
\end{lemma}
\begin{proof}
    Take $(\mathbb F_i, i)\in \mathcal F$ as defined in Section \ref{sec:prelims} and $yz\in E^{\structure F_i}$.
    Let $(\structure I, \xi^\structure I, f^\structure I)$ be the gadget provided by \Cref{lemma:step2a} for color $i\in[c]$.
    First, we remove the color from all edges incident to $y$ in $\structure F_i$.
    Then attach to it a copy of $(\structure I, \xi^\structure I, f^\structure I)$ along $f^\structure I$ to all edges incident to $y$ except $yz$ in such a way that no colored edge is incident to $y$.
    See \Cref{fig:Dnegi} for the cases $\mathbb F_i\cong \mathbb C_5$ and $\mathbb F_i\cong \mathbb K_4$.
    Call the resulting partially colored graph $(\mathbb T,\xi^\mathbb T)$.
    From now on we refer to the copy of $\structure F_i$ in $\structure{T}$ as the \emph{ground graph}.
    \begin{figure}[h]
        \begin{center}
            \includegraphics[scale=0.6]{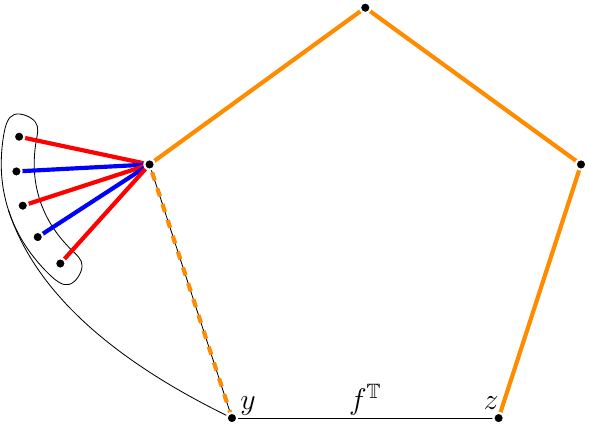}
            \hspace{1cm}
            \includegraphics[scale=0.6]{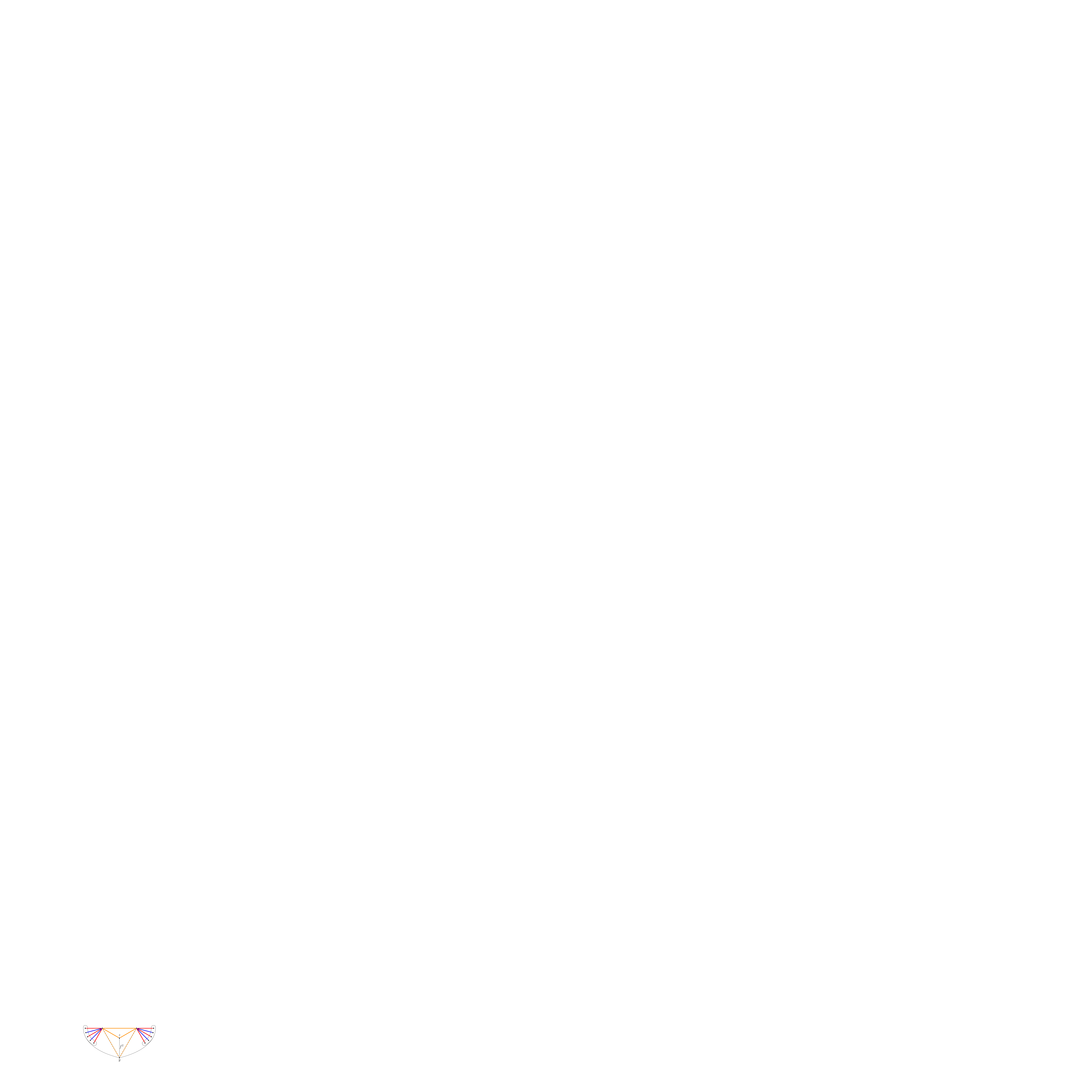}
            \caption{The construction of $\structure{T}$ in the proof of~\Cref{propositio:Dnegi}.}
            \label{fig:Dnegi}
        \end{center}
    \end{figure}
 
    Finally, we prove the desired properties.
    \Cref{item:Dnegi-distance} follows from construction, $y$ is the witness to this property.
    \Cref{item:Dnegi3} follows from the fact that all edges except $f^{\structure{T}}$ in the copy of the ground graph in $(\structure{T},\xi^{\structure{T}})$ must be colored by $i$ in every $\family F$-free extension of $\xi^{\structure{T}}$.
    Thus, any extension must color $f^{\structure{T}}$ with a color different from $i$, otherwise the ground graph is monochromatic in color $i$. 
    Item \ref{item:Dneg-soundness} follows because we only deal with monochromatic obstructions.
    \Cref{item:Dnegi4} is exactly the last item of \Cref{easyobs}. 
\end{proof}

\begin{lemma}\label{lemma:step2c}
    For every color $i\in [c]$, there exists a color-$i$-determiner.
\end{lemma}
\begin{proof}
    For $j\neq i$, let $(\structure{T}_{j},\xi^{\structure{T}_{j}},f^{\structure{T}_{j}})$ be provided by~\Cref{propositio:Dnegi} for the color $j$. 
    Let $x_j$ be the vertex of $f^{\structure{T}_{j}}$ which is not incident to any edge in $\dom{\xi^{\structure{T}_{j}}}$.
    Take the disjoint union of the colored graphs $(\structure{T}_{j} ,\xi^{\structure{T}_{j}})$ for all $i\neq j$ and identify all the edges $f^{\structure{T}_{j}}$, making sure to identify all the vertices $x_j$ together. We call this graph $\structure D_{i}$ and the edge resulting from the amalgamation is called $f^{\structure D_{i}}$.
    Clearly, every $\family F$-free extension of the coloring of $\structure D_{i}$ must color $f^{\structure D_{i}}$ with $i$, and the vertex resulting from the merging of all vertices $x_j$ does not belong to any colored edge.
    Since $\structure D_{i}$ results from the amalgamation  of $\family F$-safe structures, it is itself $\family F$-safe by~\Cref{easyobs}.
\end{proof}

\begin{restatable}{lemma}{lemmaremotedeterminers}\label{lemma:remote_determiners}
    For every $d\in\mathbb{N}$ and $i\in [c]$, there exists a $d$-remote color-$i$-determiner.
\end{restatable}
\begin{proof}
    The proof is by induction on $d$.
    The base $d=0$ is provided by \Cref{lemma:step2c}.
    For the inductive step, for every $i\in [c]$, we glue to every $i$-colored edge $f$ of a $d$-remote color-determiner a copy of a color-$i$-determiner and remove the color from $f$. 
    This remains $\family F$-safe by~\Cref{easyobs}, and is $(d+1)$-remote.
\end{proof}

\begin{restatable}{proposition}{thmprecolored}\label{lemma:precolored}
    The problem $\Ext(\family{F})$ reduces in poly-time to $\Col(\family{F})$.
\end{restatable}
\begin{proof}
    Let $d$ be the maximum number of vertices in a graph in $\family{F}$.
    Let $(\structure X,\xi^\mathbb X)$ be a partially colored graph.
    Let $\mathbf T_i:=(\structure T_i, \xi^{\mathbb T_i}, f^{\mathbb T_i})$ be a $d$-remote color-$i$-determiner provided by \Cref{lemma:remote_determiners}.
    For every $i\in[c]$ and for every edge $e\in\dom{\xi^\mathbb X}$ colored with $i$, we glue a copy of $\mathbf T_i$ by identifying $e$ with $f^{\mathbb T_i}$, then make $e$ uncolored.
    For every $i\in[c]$ and for every two copies $(\structure T_i, \xi^{\mathbb T_i}, f^{\mathbb T_i})$ and $(\structure T_i', \xi^{\mathbb T_i'}, f^{\mathbb T_i'})$ of $\mathbf T_i$, we identify every edge of $\dom{\xi^{\mathbb T_i}}\subset E^{\mathbb T_i}$ with its copy in $\dom{\xi^{\mathbb T_i'}}\subset E^{\mathbb T_i'}$.
    This makes $\xi^{\mathbb T_i}$ and $\xi^{\mathbb T_i'}$ the same mapping $\xi_i$.

    Denote the resulting structure by $(\structure X_1, \xi^{\mathbb X_1})$, where $\xi^{\mathbb X_1} = \bigoplus_i \xi_i$.
    Let $\bar c \in (X_1)^\ell$ be the tuple of all vertices that are incident to edges in $\dom{\xi^{\mathbb X_1}}$.
    As every color-$i$-determiner $\mathbf T_i$ has a $\mathcal F$-free extension, there is a homomorphism $g_i\colon \mathbb T_i \to \mathbb G_\mathcal F$ such that $\chi^{\mathbb C}(g_i(e))=\xi^{\mathbb T_i}(e)$ for every $e\in\dom{\xi^{\mathbb T_i}}$, where $\mathbb C = (\mathbb G_\mathcal F,\chi^\mathbb C)$ is the $\sigma$-expansion of $\mathbb G_\mathcal F$.
    Let $g\colon \bigoplus_i\mathbb T_i \to \mathbb G_\mathcal F$ be such $g(x)=g_i(x)$ whenever $x\in T_i$, for every $i\in[c]$.
    Let $(\mathbb D, \bar d)$ be the pp-definition of the $\Aut(\mathbb G_\mathcal F)$-orbit of $g(\bar c)$ provided by \Cref{cor:mcc_pp}, and let $h\colon \mathbb D\to \mathbb G_\mathcal F$ be such that $g(\bar c)=h(\bar d)$.
    Glue  $\structure D$ to $\structure X_1$ by identifying the vertices of  $\bar d$ and $\bar c$ component-wise, remove  $\xi^{\mathbb X_1}$, and call the resulting graph $\structure Y$, see \Cref{fig:reduction}.

    Let $(\mathbb X,\chi^\mathbb X)$ be a $\mathcal F$-free extension of $\xi^\mathbb X$ and $\chi^{\mathbb Y}\colon E^{\mathbb Y}\to [c]$ be equal to $\chi^\mathbb X$ when restricted on $E^\mathbb X$ and the edge-colors of the rest are provided by the images of $g$ and $h$.
    The fact that $(\mathbb Y, \chi^{\mathbb Y})$ is $\mathcal F$-free follows from the $\mathcal F$-safeness of color-$i$-determiners.

    Let $\chi^{\mathbb Y}$ be an $\mathcal F$-free edge-coloring of $\mathbb Y$, then there is a homomorphism $s\colon \mathbb Y \to \mathbb G_\mathcal F$.
    Let $\alpha \in\Aut{\mathbb G_\mathcal F}$ such that the edges of $\alpha(s(\bar c))$ have the same colors as in $\xi^{\mathbb X_1}$.
    Then the edge-coloring of $\mathbb Y$ provided by $\alpha\circ s$ will extend $\xi^\mathbb X$.
\end{proof}
    
\begin{figure}
    \centering
    \includegraphics[width=\linewidth]{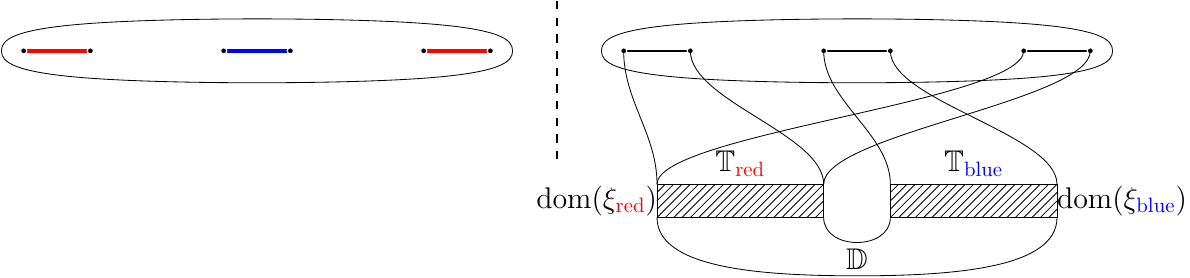}
    \caption{Reduction from \Cref{lemma:precolored}: structure $\mathbb X$ is on the left and $\mathbb Y$ is on the right.}
    \label{fig:reduction}
\end{figure}

\subsection{Reduction from an NP-hard finite-domain CSP to Ext}

\begin{definition}\label{equality-gadget}
    A partially colored graph $(\structure H^{=}, \xi^=)$ with distinguished edges $e^=, f^=\in E^{\mathbb H^=}\setminus \dom{\xi^=}$ is a \emph{color-equality gadget} for $\family F$ if it satisfies the following:
    \begin{enumerate}[(1)]
        \item\label{item:3eq} every $\family F$-free extension of $\xi^=$ assigns to $e^=$ and $f^=$ the same color,
        \item\label{item:2eq} for every $i\in [c]$, there exists an $\family F$-free extension $\xi$ of $\xi^=$ such that $\xi(e^=)=\xi(f^=)=i$.
        \item\label{item:4eq} $(\structure H^=, \xi^=, e^=, f^=)$ is $\family F$-safe.
    \end{enumerate}
\end{definition}

Recall from Section \ref{sec:prelims} that we indicate $\structure F_p$ the $\rightarrow$-minimum graph such that $(\structure F_p, p) \in \family{F}$, where, abusing notation, $p$ indicates the constant coloring whose image is $\{p\}$.
From now until Lemma~\ref{lemma:eqund} we fix $i,j\in [c]$.
\begin{lemma}\label{lemma:step1eq}
    There is a partially colored graph $(\structure H, \xi^{\structure H})$ and $f^{\structure H}\in E^\mathbb H\setminus \dom{\xi^\mathbb H}$ such that:
    \begin{enumerate}[(1)]
        \item every $\mathcal F$-free extension of $\xi^{\structure H}$ assigns to $f^{\structure H}$ either $i$ or $j$,
        \item for every $p\in \{i,j\}$ there exists an $\mathcal F$-free extension of $\xi^{\structure H}$ that assigns to $f^{\structure H}$ the color $p$,
        \item $(\structure H, \xi^{\structure H}, f^{\structure H})$  is $\family F$-safe.
    \end{enumerate}
\end{lemma}
\begin{proof}
    Let us color all the edges of $\structure F_p$ in $p$ except for $e\in E^{\mathbb F_p}$ and denote the result $(\structure F_p, p\setminus e)$.
    $(\structure F_p, p\setminus e, e)$ is $\family F$-safe by~\Cref{easyobs}.
    Let $\mathbb H$ be obtained from the disjoint union of all $(\structure F_p, p\setminus e, e)$, for $p\not\in\{i,j\}$, by identifying the uncolored edges into one called $f^{\structure H}$.
    Since each $(\structure F_p, p\setminus e, e)$ is $\family F$-safe, $(\structure H, \xi^{\structure H}, f^{\structure H})$ is $\mathcal F$-safe as well, by~\Cref{easyobs}.
    Moreover, $f^{\structure H}$ cannot take any color other than $i$ or $j$ in any $\family F$-free extension of the coloring $\xi^\mathbb H$.
    Finally, assigning $i$ or $j$ to $f^{\structure H}$ does not create an obstruction.
\end{proof}

\begin{lemma}\label{lemma:iarrowj}
    There is a partially colored graph $(\structure U, \xi^{\structure U})$ and $e_{i}, e_{j}\in E^\mathbb U\setminus \dom{\xi^{\structure U}}$ such that:
    \begin{enumerate}[(1)]
        \item\label{itm:ij} for every $\mathcal F$-free extension $\gamma$ of $\xi^{\structure U}$, $\gamma(e_i)= i$ implies $\gamma(e_j) = j$,
        \item\label{itm:ij2} for each $p\in [c]\setminus\{i\}$ and $q\in\{i,j\}$ there exists an $\mathcal F$-free  extension $\gamma$ of $\xi^{\structure U}$ with $\gamma(e_i)=p$ and $\gamma(e_j)=q$,
        \item both $(\structure U, \xi^{\structure U}, e_{i})$ and $(\structure U, \xi^{\structure U}, e_{j})$ are $\family F$-safe.
    \end{enumerate}
\end{lemma}
\begin{proof}
    We color all except two edges of $\structure F_i$, we call the two uncolored edges $e_i$ and $e_j$.
    If possible, take them to be at a distance larger than $0$. 
    Notice that this is always possible, unless $\structure F_i = \structure K_3$.
    Take $(\structure H,\xi^{\structure H},f^{\structure H})$ provided by Lemma~\ref{lemma:step1eq} and identify $e_j$ with $f^{\structure H}$.
    We indicate the result as $(\structure U, \xi^{\structure U})$. 
    If $e_i$ is colored in $i$, then $e_j$ cannot be colored in $i$.
    Since $e_j$ can only be colored in $i$ or $j$ by the properties of $(\structure H,\xi^{\structure H}, f^{\structure H})$, we get that it can only be colored in $j$.
    This proves item \ref{itm:ij}.
    Item \ref{itm:ij2} follows immediately since we only have monochromatic obstructions.
    The last item follows from \Cref{easyobs}.
\end{proof}

Now,  for every 2-element subset of colors $\{i,j\}$, we build a graph with two distinguished edges such that we can assign to them every pair of colors but $i,j$.

\begin{lemma}\label{lemma:Gd}
    For every $d\geq 1$, there is a partially colored graph $(\structure G_d, \xi^{\structure G_d})$ and $e_{i}, e_{j}\in E^{\structure G_d}\setminus\dom{\xi^{\structure G_d}}$ such that:
    \begin{enumerate}[(1)]
        \item\label{itm:notij} for every $\mathcal F$-free extension $\gamma$ of $\xi^{\structure G_d}$, either $\gamma(e_i)\neq i$ or $\gamma(e_j)\neq j$,
        \item\label{itm:allnotij} for every pair of colors $(p,q)$ different from $(i,j)$, there exists an $\mathcal F$-free extension $\gamma$ of $\xi^{\structure G_d}$ with $\gamma(e_i)=p$ and $\gamma(e_j)=q$,
        \item both $(\structure G_d, \xi^{\structure G_d},e_{i})$ and $(\structure G_d, \xi^{\structure G_d},e_{j})$ are $\family F$-safe.
        \item $\dist(e_i,e_j)\geq d$.
    \end{enumerate}
    Moreover, if $d$ is greater than the size of every graph in $\mathcal F$, then
    \begin{enumerate}[(1)]
    \setcounter{enumi}{4}
    \item\label{item:last} $(\structure G_d, \xi^{\structure G_d}, e_i,e_j)$ is $\mathcal{F}$-safe.
    \end{enumerate}
\end{lemma}

\begin{proof}
    We first examine the case $d=1$. For the moment, we assume that $\structure F_j \neq \structure K_3$ or $\structure F_i \neq \structure K_3$. We consider without loss of generality that $\structure F_j \neq \structure K_3$.
    We define a graph $\structure G^{i,j}_1$ as follows.
    Take $(\structure U, \xi^{\structure U}, h_i, h_j)$ from Lemma~\ref{lemma:iarrowj} and $(\structure F_j, j\setminus e^{\structure F_j} f^{\structure F_j})$ where $e^{\structure F_j}, f^{\structure F_j}$ are not adjacent to each other. Notice that this is always possible since $\structure F_j \neq \structure K_3$.
    Now we define $(\structure G_1, \xi^{\structure G_1} , e_i, e_j)$ by taking two copies of partially colored graphs just mentioned and identifying $h_j$ with $f^{\structure F_j}$, and renaming $e_i$ the edge $h_i$ and $e_j$ the edge $e^{\structure F_j}$.
    
    We start proving the desired properties of the structure obtained.
    By definition of $(\structure F_j, j\setminus e^{\structure F_j} f^{\structure F_j})$ for every $\family F$-free valid extension $\gamma$ of $j\setminus e^{\structure F_j} f^{\structure F_j}$ we have $\gamma(e^{\structure F_j})\neq j$ or $\xi(f^{\structure F_j})\neq j$ otherwise $\structure F_j$ would be monochromatic in $j$.
    By the properties of $(\structure U, \xi^{\structure U}, h_i, h_j)$ we know that if $e_i$ takes color $i$ then $h_j \equiv f^{\structure F_j} $ must be colored in $j$ then by definition of $(\structure F_j, j\setminus e^{\structure F_j} f^{\structure F_j})$ we get that $e_j \equiv e^{\structure F_j}$ cannot be colored in $j$. This proves Item \ref{itm:notij}.
    For Item \ref{itm:allnotij} it suffices to see that if $e_i$ is not colored in $i$ then $h_j \equiv f^{\structure F_j} $ can be colored in $i$ hence we get that $e_j$ can be colored in all colors.
    Similarly if $e_i$ is colored in $i$ then $h_j \equiv f^{\structure F_j} $ is colored in $j$ but $e_j$ can be colored in all colors except $j$.
    Item 3 follows by \Cref{easyobs}.
    The last item is ensured by the choice of $e^{\structure F_j}, f^{\structure F_j}$, since they are not adjacent.
    So it only remains to prove that from this we can build a gadget where the distance between the edges is arbitrarily large and that, if it is large enough, it allows us to glue both edges simultaneously in an $\family F$-safe way.
    Before that, we give a graphical representation of  $\structure G_1$, in the case where the obstruction set consists of monochromatic $\structure C_5$ and $\structure K_3$ for the colors: red, blue, pink and orange.
    \begin{center}
        \includegraphics[scale=0.7]{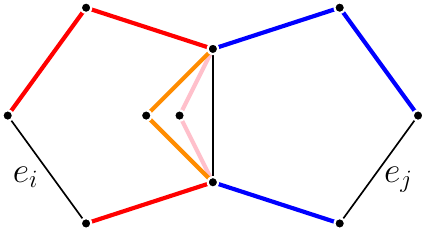}
    \end{center}
    Define now $\structure G_d$ by gluing $d$ copies of $(\structure G_1, \xi^{\structure G_1})$ in sequence, in a way that the edge $e_j$ of a copy is glued to the edge $e_i$ of the next copy.
    In addition, we attach $\structure H$, from Lemma~\ref{lemma:step1eq}, to all uncolored edges in the graph except for the first $e_i$ and the last $e_j$ (see~\Cref{fig:Gd}).
    To prove~\Cref{itm:notij}, suppose that $\gamma$ is an extension of $\xi^{\structure G_d}$ with $\gamma(e_i)=i$ and $\gamma(e_j)=j$. Then there must be some copy of $\structure G_1$ in the sequence whose edges $e_i,e_j$ are also colored similarly, and therefore $\gamma$ is not $\family F$-free. 
    \Cref{itm:allnotij} is clear, as the middle edges of the sequence can be colored with $i$ or $j$ depending on whether $p\neq i$ or $q\neq j$.
    The fact that $(\structure G_d, \xi^{\structure G_d}, e_i)(\structure G_d, \xi^{\structure G_d}, e_j)$ are $\family F$-safe follows from~\Cref{easyobs}, and the distance between $e_i$ and $e_j$ is clearly at least $d$.
    To prove~\Cref{item:last}, consider $(\structure T ,\chi^\structure T) \in \forb{\family F}$ with two edges $e,f\in E^\structure T$ such that there exists an extension $\gamma$ of $\xi^{\structure G_d}$, and $\gamma (e_i) = \chi^\structure T(e), \gamma(e_j) = \chi^\structure T(f)$.
    We take a copy of $(\structure T ,\chi^\structure T)$ and a copy of $(\structure G_d, \gamma)$ and identify $e$ with $e_i$ and $f$ with $e_j$, calling the result $(\structure T \oplus \structure G_d, \chi^\structure T \oplus \gamma)$. Suppose that there is a homomorphism from $(\structure F,\chi^{\structure F})\in\family F$ to $(\structure H \oplus \structure G_d,\chi^{\structure T} \oplus \gamma)$.
    Since $d$ is larger than the size of the obstructions, the intersection of the image of $h$ with the sequence of $\structure G_1$ cannot induce a connected graph; thus, $h$ is a homomorphism to the amalgam $\structure G_d\oplus\structure T\oplus \structure G_d$ obtained from a copy of $\structure T$ and two copies of $(G_d)^1, (G_d)^2$ identifying $e_i^1$ with $e$ and $e_j^2$ with $f$.
    By the safeness of $e_i^1$ and $e_j^2$, we get $(\structure G_d\oplus\structure T\oplus \structure G_d , \gamma \oplus \chi^{\structure T} \oplus \gamma) \in \forb{\family F}$ as desired.
    It only remains to explain how to get the gadget for $d=1$ in the case $\structure F_i= \structure K_3 = \structure F_j$.
    It suffices to glue them as shown in the following picture:
    \begin{center}
        \includegraphics[scale=0.6]{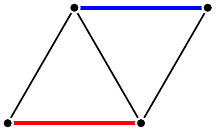}
    \end{center}
    and attach on the middle edge the gadget $(\structure H, \xi^{\structure H}, f^\structure H)$ given by Lemma~\ref{lemma:step1eq}.
    Here, $e_i, e_j$ are the leftmost and rightmost edges, respectively.
\end{proof}

\begin{figure}[h]
    \centering
    \includegraphics[width=0.5\textwidth]{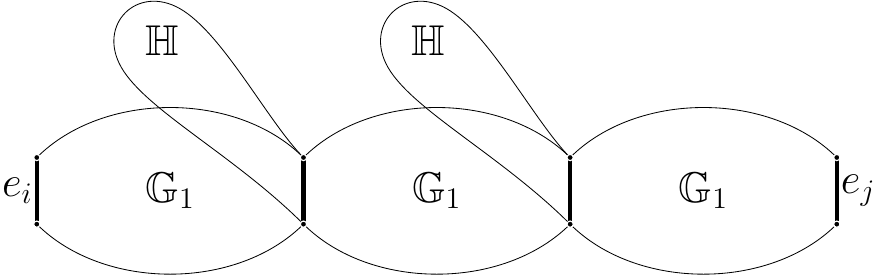}
    \caption{Construction of $\structure G_3$ from $\mathbb G_1$ from~\Cref{lemma:Gd}.}    
    \label{fig:Gd}
\end{figure}

Finally, we have all the ingredients to build a color-equality gadget.

\begin{proposition}\label{lemma:eqund}
There exists a color-equality gadget for $\family F$.
\end{proposition}
\begin{proof}
    For every distinct $i,j\in [c]$, introduce a copy of $(\structure G_d^{ij}, \xi^{\structure G_d^{ij}}, e_i, e_j)$ from \Cref{lemma:Gd} and identify first their $e_i$-edges and then identify their $e_j$-edges.
    The fact that the resulting graph $(\structure H^=,\xi^{=},e^=,f^=)$ is $\family F$-safe can be proved by applying \Cref{easyobs} finitely many times, since $(\structure H^=,\xi^{=})$ is the result of amalgamating graphs on pairs of edges along which they are safe. 
    Every $\family F$-free extension $\xi$ of $\xi^=$ must assign $e^=$ and $f^=$ to the same color, as $(\xi(e^=),\xi(f^=))$ must be different from $(i,j)$ for every $i\neq j$.
\end{proof}

Let $\mathbb A_\mathcal F$ be the relational structure with domain $A = [c]$, the constant $\{i\}$ for each $i\in [c]$, and a $|E^\structure F|$-ary relation $R_{\structure F}$, for each odd cycle or a clique $\mathbb F$ with an edge-coloring in $\mathcal F$, containing exactly those $\bar a\in A^{E^\mathbb F}$ where $(\mathbb F, \bar a)$ is $\mathcal F$-free.

\begin{proposition}\label{lemma:prec-csp}
    For every $\mathcal F$ which has a color-equality gadget, there is a poly-time reduction from $\CSP(\structure A_{\family{F}})$ to $\Ext(\family{F})$.
\end{proposition}
\begin{proof} 
    Let $\structure X$ be an input for $\CSP(\structure A_{\family{F}})$.
    We define a partially colored graph $\structure H_{\structure X}$ that is a yes-instance for $\Col(\family{F})$ if and only if $\structure X$ is a yes-instance for $\CSP(\structure A_{\family{F}})$.
    We define $\structure H_{\structure X}$ in steps. 
    First, define $\structure T_{\structure X}$ as the disjoint union of copies of $\structure F$, one for each tuple in $(R_{\structure F})^{\structure X}$, for all relations in the signature of $\structure A_{\family{F}}$.
    Recall that each element in an instance of $R_{\structure F}$ corresponds to a specific edge of $\structure F$. 
    In the second step we add a copy of the color-equality gadget, for each edge in $\structure T_{\structure X}$ that corresponds to the same vertex in $\structure X$.
    We glue $e^=$ to one edge and $f^=$ to the other. 
    In the third step we color with $i$ all the edges that correspond to the element $x$ if $x$ is assigned to $i$. We call the result $\structure H_{\structure X}$. Suppose now that it is $\family{F}$-free colorable. 
    Then we define a homomorphism from $\structure X$ to $\structure A_{\family{F}}$ as $x \mapsto i$ if one edge corresponding to $x$ is colored in $i$.
    Notice that by the properties of the gadget that we built this is well-defined.
    Moreover, it is immediate that this map is a homomorphism for how we defined $\structure A_{\family{F}}$.
    Suppose now that $\structure X$ has a homomorphism to $\structure A_{\family{F}}$ then define a coloring on $\structure H_{\structure X}$ such that an edge is colored in $i$ if and only if the corresponding element is assigned to $i$ by the homomorphism. 
    It is immediately seen that any copy of obstructions in $\structure H_{\structure X}$ has a valid coloring according to  the definition of $\structure A_{\family F}$.
    This concludes the argument.   
\end{proof}

We finally turn to the NP-completeness of $\CSP(\structure A_{\family F})$.
A \emph{polymorphism} of a relational structure $\structure{A}$ is a function $A^k\to A$ for some $n\geq 1$ such that for every relation $R\subseteq A^\ell$ of $\structure A$ and all tuples $a^1,\dots,a^k\in A$, the tuple $(f(a^1_1,\dots,a^k_1),\dots,f(a^1_\ell,\dots,a^k_\ell))$ is in $R$.
A function $f\colon A^k \rightarrow A$, for $k \geq 2$, is called \emph{cyclic} if it satisfies
$\forall x_1, \ldots, x_k\in A: f\left(x_1, \ldots, x_k\right)=f\left(x_2, \ldots, x_k, x_1\right)$.
It is known~\cite[Theorem 4.1]{Absorption} that if $\structure A$ is finite and has no cyclic polymorphism, then $\CSP(\structure A)$ is NP-complete.

\begin{proposition}\label{prop:cspNP}
    $\CSP(\structure A_{\family{F}})$ is NP-complete.
\end{proposition}
\begin{proof}
    We prove by contradiction that $\structure A_{\family{F}}$ does not have any cyclic polymorphism of any arity.
    This implies our claim by~\cite[Theorem 4.1]{Absorption}.
    Assume that $\structure A_{\family F}$ has a cyclic polymorphism of arity $k\geq 2$.
    Take $m$ to be $k/2$ if $k$ is even and $(k-1)/2$ if $k$ is odd.
    Take two different colors $i,j$ and let $p$ be $f(i,\ldots,i,j,\ldots,j)$, where the number of $i$ is $m$.
    Note that $f(p,\dots,p)=p$, since the singleton unary relation $\{p\}$ is a relation of $\structure A_{\family F}$.
    The coordinate of $R_{\structure F_p}$ are in one-to-one correspondence with $E^{\structure F_p}$, without loss of generality we assume that the first three coordinate do not correspond to a $\structure K_3$ under this correspondence. 
    Then $\forall i\neq p, j\neq p, q\neq p$ we have $(i,j,p,\ldots,p), (i,j,q,p,\ldots,p) \in R_{\structure F_p}$ and $(p,\dots,p)\not\in R_{\structure F_p}$. 
    We let $\ell\geq 3$ be the arity of $R_{\structure F_p}$.
    Consider the following $(\ell\times k)$-matrix on the left if $k$ is even, and on the right if $k$ is odd.
    \[
    \begin{tikzpicture}[baseline=(m-4-1.base)]
    \matrix (m) [matrix of math nodes,left delimiter={(},right delimiter={)}] {
    i & i & \dots & i & j & j& \dots & j\\
    j & j& \dots & j & i & i& \dots & i\\
    p & p &\dots  & p & p& p & \dots &p\\
    p & p &\dots  & p & p& p & \dots &p\\
    \vdots & \vdots & & \vdots & \vdots & \vdots & & \vdots\\
    p & p &\dots  & p & p& p&\dots &p \\
    };

    \draw[<->] (m-1-1.north west) -- node[above]{$m$} (m-1-4.north east);
    \draw[<->] (m-1-5.north west) -- node[above]{$m$} (m-1-8.north east);
    \end{tikzpicture}
    \quad\text{or}\quad
    \begin{tikzpicture}[baseline=(m-4-1.base)]
    \matrix (m) [matrix of math nodes,left delimiter={(},right delimiter={)}] {
    i & i & \dots & i & j & j& \dots & j\\
    j & j& \dots & j & j & i& \dots & i\\
    j & i & \dots & i & i & j & \dots & j\\
    p & p &\dots  & p & p& p & \dots &p\\
    \vdots & \vdots & & \vdots & \vdots & \vdots & & \vdots\\
    p & p &\dots  & p & p& p&\dots &p \\
    };

    \draw[<->] (m-1-1.north west) -- node[above]{$m$} (m-1-4.north east);
    \draw[<->] (m-1-6.north west) -- node[above]{$m$} (m-1-8.north east);
    \end{tikzpicture}
    \]
    Applying $f$ to these rows gives the constant $\ell$-tuple $(p,\dots,p)\not\in R_{\structure F_M}$, a contradiction to the fact that $f$ is a polymorphism of $\structure A_{\family F}$ since all the columns in the matrix are in $R_{\structure F_p}$.
\end{proof}

Consequently applying \Cref{lemma:precolored,lemma:prec-csp,prop:cspNP} we prove \Cref{thm:main}.

\section{Inexpressibility results for forbidden vertex colorings}\label{sec:inexpressibility}

Let $\family F$ be a family of \emph{vertex-}colored graphs.
$\Col(\family F)$ denotes as above the problem of deciding whether a given input graph admits a vertex-coloring that is $\family F$-free.
Without loss of generality, we can assume that all the colored graphs in $\family F$ have domain of the form $\{1,\dots,k\}$.
Define $\structure H_{\family F}$ to be the structure with domain the set of colors $[c]$, and with a relation $R_{\structure F}$ of arity $k$ for every graph that appears in $\structure F$ under some coloring.
The relation $R_{\structure F}$ contains all the tuples $(n_1,\dots,n_k)$ such that the coloring $\chi$ of $\structure F$ defined by coloring $i$ with the color $n_i$ is $\family F$-free.

There is a standard reduction from $\Col(\family F)$ to $\CSP(\structure H_{\family F})$ which, given an input $\structure G$, returns a new structure $\structure H$ whose domain is the set of vertices of $\structure G$, and such that a tuple $(v_1,\dots,v_k)$ is in a relation $R_{\structure F}$ of $\structure H$ if and only if the map $i\mapsto v_i$ is a homomorphism $\structure F\to \structure G$.

In an arbitrary relational structure $\structure H$, a \emph{cycle} of length $k$ is a sequence of tuples $\bar t_1,\dots,\bar t_k$, where each $\bar t_i$ is an element of a relation $ R_i$ of $\structure H$ of arity $r_i$, and such that the total number of elements appearing in those tuples is at most $\sum_{i=1}^k (r_i-1)$.
Note that two tuples of arity $\geq 2$ and having $\geq 2$ elements in common form a cycle of length $2$. 
A \emph{tree} is a structure without cycles.
The \emph{girth} of $\structure H$ is the smallest length of a cycle in $\structure H$.
The \emph{sparse incomparability lemma}~\cite[Theorem 5]{FederVardi} states that for all $g\geq 1$, if $\structure H,\structure K$ are finite structures such that $\structure H\not\to \structure K$, then there is a structure $\structure H'$ with girth at least $g$, such that $\structure H'\to \structure H$ but $\structure H'\not\to \structure K$.

\separation*
\begin{proof}
    Suppose for contradiction that such a family $\family F$ exists, and let $K$ be the maximal size of the colored graphs in $\family F$.
    By the second item in the statement, there exists a graph $\structure G$ that does not have an $\family E$-free edge coloring (equivalently, an $\family F$-free vertex coloring) and such that every subgraph of $\structure G$ of size at most $K$ has an $\family E$-free coloring.

    Let $\structure H$ be the structure obtained by the reduction from $\Col(\family F)$ to $\CSP(\structure H_{\family F})$ described above applied to $\structure G$.
    This $\structure H$ does not have a homomorphism to $\structure H_{\family F}$.
    By the sparse incomparability lemma, there exists $\structure H'$ of girth $g>\max\{|E| \mid \structure E\in\family E\}$ with a homomorphism to $\structure H$ and no homomorphism to $\structure H_{\family F}$.

    Let $\structure G'$ be the following graph: its vertices are the elements of $\structure H'$, and for every $(v_1,\dots,v_k)$ in a relation $R_{\structure F}$ in $\structure H'$, we add to $\{v_1,\dots,v_k\}$ the edges needed for the map $F\to\{v_1,\dots,v_k\}$ that sends $i$ to $v_i$ to be a homomorphism from $\structure F$ to $\structure G'$.
    Then note that in the reduction from $\Col(\family F)$ to $\CSP(\structure H_{\family F})$, $\structure G'$ is mapped to a structure containing every tuple of $\structure H'$ and therefore is a no-instance of $\CSP(\structure H_{\family F})$, from which it follows that $\structure G'$ is a no-instance of $\Col(\family F)$.
    The homomorphism $h\colon \structure H'\to \structure H$ is also a homomorphism $\structure G'\to \structure G$.
    Indeed, an edge $(v_i,v_j)$ in $\structure G'$ arises from a tuple $(v_1,\dots,v_k)\in R_{\structure F}$ in $\structure H'$, i.e., from a homomorphism $g'\colon \structure F\to \structure G'$ such that $(i,j)$ is an edge in $\structure F$ and $g'(i)=v_i,g'(j)=v_j$.
    Thus, $\bigl(h(v_1),\dots,h(v_k)\bigr)\in R_{\structure F}$ in $\structure H$ and therefore there is a homomorphism $g\colon \structure F\to \structure G$ with $g(i)=h(v_i)$ and $g(j)=h(v_j)$.
    It follows that $h(v_i),h(v_j)$ are connected by an edge in $\structure G$.

    We claim that every subgraph of $\structure G'$ induced by vertices $\{v_1,\dots,v_k\}$ where $(v_1,\dots,v_k)$ is in a relation $R_{\structure F}$ in $\structure H'$ admits an $\family E$-free edge coloring.
    Indeed, first notice that $k=|F|\leq K$.
    If the graph induced by $\{v_1,\dots,v_k\}$ does not admit an $\family E$-free coloring, then its image under $h$ is a subgraph of $\structure G$ of size at most $K$ that does not admit an $\family E$-free coloring, a contradiction.

    Every edge of $\structure G'$ is contained in a unique such maximal subgraph $\{v_1,\dots,v_k\}$ arising from a tuple in $\structure H'$. 
    Define a coloring $\chi'$ of all the edges of $\structure G'$ by gluing together all the colorings obtained in the previous paragraph.
    We claim that this coloring is $\family E$-free, which gives us the desired contradiction.
    Suppose that there is $(\structure E,\theta)\in\family E$ that admits a homomorphism $f$ to $(\structure G',\chi')$.
    Since $\structure E$ has size at most $g$, the image of $f$ in $\structure H'$ induces a tree structure.
    Since the class of $\family E$-free colored graphs is closed under vertex amalgamations by assumption, it follows that the image of $f$ must be included in a single edge of $\structure H'$.
    However, this corresponds to a subgraph of $\structure G'$ that we proved above to be $\family E$-free, and this concludes the proof.
\end{proof}

\section{Conclusion}\label{section:conclusion}

Our hardness proof depends on the existence of color-determiners for the corresponding families of colored graphs.
This is currently an important obstacle to generalizing our results; indeed, we know that some families of graphs do not admit color-determiners. We give an example in Appendix~\ref{section:picNP}, where we also show by ad hoc methods that the coloring problem $\Col(\family F)$ is NP-complete.
The second step of the proof using color-equality gadgets could in principle still apply, using the relaxed notion of \emph{smooth approximation} from~\cite{SmoothApproximations}; namely, we do not know of any example where $\Ext(\family F)$ is not poly-time equivalent to the finite-domain CSP described above.

\bibliographystyle{plainurl}
\bibliography{references}

\appendix

\section{Color-determiners do not always exist}\label{section:picNP}
Let $\sigma\colon \{1,2\}\to\{1,2\}$ be the transposition of 1 and 2.
Let $\mathrm{sw}_\sigma$ be a function that takes as input a pair: a 2-edge-colored graph $(\mathbb G, \chi)$ and a vertex $v\in G$, and returns a 2-edge-colored graph $\mathrm{sw}_\sigma(\mathbb G, \chi, v)$ which is obtained from $\mathbb G$ by switching the colors of the edges incident to $v$ to the opposite ones.
Say that a family $\mathcal F$ of 2-edge-colored graphs is \emph{closed} under $\mathrm{sw}_\sigma$ if $\mathrm{sw}_\sigma(\mathbb F, \phi, v)\in \mathcal F$ for every $(\mathbb F, \phi)\in \mathcal F$ and every $v\in F$.
Note that $\mathcal F$ being closed under $\mathrm{sw}_\sigma$ implies that color-determiners associated with $\Col(\mathcal F)$ do not exist: as partially colored edges must share no common vertices with the edge whose color is determined, applying $\mathrm{sw}_\sigma$ to any of the two vertices incident to this edge preserves the partial coloring of the determiner.
\begin{figure}[ht]
    \centering
    \includegraphics[width=0.35\linewidth]{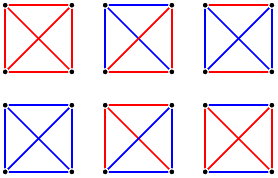}
    \caption{An example of a  family which is closed under $\mathrm{sw}_\sigma$.}
    \label{fig:counterexample}
\end{figure}
Let $\mathcal F$ be the family of 2-edge-colored copies of $\mathbb K_4$ as in Figure~\ref{fig:counterexample}.
Clearly, $\mathcal F$ is closed under $\mathrm{sw}_\sigma$ and thus has no color-determiners.
On the bright side, the non-existence of color-determiners does not imply the non-existence of a P vs.\ NP-complete complexity dichotomy.
Transformations of the form $\mathrm{sw}_\sigma$ already appear while studying universal $\omega$-categorical structures defined by forbidden tournaments under the name of \emph{switches}~\cite{Thomas,Kompatscher,Santiago}.
This, however, was not an obstacle in proving a dichotomy for such problems, although it added much technical work.
We believe that it should not be an obstacle for the edge-coloring problems either.
For the family $\mathcal F$ from \Cref{fig:counterexample}, the problem $\Col(\mathcal F)$ is NP-complete; the rest of this section is dedicated to proving this.
We denote by $\family{F}_{\mathrm{sw}}$ the family in \Cref{fig:counterexample}. 
We prove that $\Col(\family{F})$ is NP-complete. 
Before doing that, we need to introduce a few concepts.
Given a graph $\structure G$, we define $\structure G*v_{\structure G}$ as follows. 
$\structure G*v_{\structure G}$ is obtained from $\structure G$ by adding a new vertex and connecting it through an edge to all other vertices.
Notice that $\structure G$ is a subgraph of $\structure G*v_{\structure G}$.
Moreover, given a coloring $\chi$ of $\structure G$ we define a coloring of the edges of $\structure G*v_{\structure G}$, $\chi *_i v_{\structure G}$ as $\chi$ on the subgraph $\structure G$ and as $i$ on the edges incident to $v_\mathbb G$.
We now introduce an auxiliary family of colored graphs that is central in the complexity characterization of $\family{F}_{\mathrm{sw}}$.
\begin{definition}
    For a fixed family $\family F$ and a color $i\in [2]$ we define the family $\family F^{*i}$ as follows. 
    Add an edge-colored graph $(\structure F, \alpha)$ to $\family F^{*i}$ if either $(\structure F* v_{\structure F}, \alpha *_i v_{\structure F}) \in \family F$ or $(\structure F, \alpha) \in \family F$.
\end{definition}
We now prove a very useful proposition.
\begin{proposition}\label{proposition:F*-F}
    $(\structure G,\chi) \in \forb{\family F_{\mathrm{sw}}^{*i}}$ if and only if $(\structure G*v_{\structure G},\chi*_i v_{\structure G}) \in \forb{\family F_{\mathrm{sw}}}$.
\end{proposition}
\begin{proof}
    First, the direction from left to right follows from the definition of $\family F_{\mathrm{sw}}^{*i}$.
    Secondly, we prove the direction from right to left by contraposition. So we assume that $\exists (\structure F, \phi) \in \family F_{\mathrm{sw}}^{*i}$ and a homomorphism $h: (\structure F, \phi) \rightarrow (\structure G,\chi)$.
    If $(\structure F, \phi) \in \mathcal F$, then we are done.
    Otherwise, $(\structure F* v_{\structure F}, \phi*_i v_{\structure F}) \in \mathcal F$ and there exists a homomorphism $h':(\structure F* v_{\structure F}, \phi*_i v_{\structure F}) \rightarrow (\structure G*v_{\structure G},\chi*_i v_{\structure G})$ which is defined as $h$ on $\structure F$ and as $v_{\structure F} \mapsto v_{\structure G}$.
\end{proof}
We now prove the existence of a polynomial-time reduction from $\Col(\family F_{\mathrm{sw}}^{*i})$ to $\Col (\family F_{\mathrm{sw}})$.
\begin{lemma}\label{lemma:*ileq}
    $\Col(\family F_{\mathrm{sw}}^{*i}) \leq_p \Col (\family F_{\mathrm{sw}})$.
\end{lemma}
\begin{proof}
    Take $\structure G$ as an input for $\Col(\family F_{\mathrm{sw}}^{*i})$. 
    We want to show that $\structure G$ is a yes-instance for $\Col(\family F_{\mathrm{sw}}^{*i})$ if and only if $\structure G*v_{\structure G}$ is a yes-instance for $\Col(\family F_{\mathrm{sw}})$.
    This immediately implies the result, since $\structure G*v_{\structure G}$ can be computed from $\structure G$ in poly-time.
    Suppose that $\structure G$ is a yes-instance witnessed by $\chi$. 
    By Proposition \ref{proposition:F*-F} we immediately see that $\structure G*v_{\structure G}$ is a yes-instance witnessed by $\chi*_iv_{\structure G}$.
    Now we prove the other direction. Suppose that $\structure G*v_{\structure G}$ is a yes-instance witnessed by $\xi$. 
    We define a new coloring $\xi'$ in $\structure G*v_{\structure G}$ by applying switches to specific vertices.
    The obtained colored graph is still in $\forb{\family F_{\mathrm{sw}}}$.
    To be specific, we define $\xi':= \mathrm{sw}_{\sigma}(\ldots(\mathrm{sw}_{\sigma} (\structure G*v_{\structure G}, \xi, u_1)\ldots,u_n)$ where $\{u_1,\ldots,u_n\}$ is the set of vertices of $\structure G*v_{\structure G}$ such that $\forall l\in \{1,\ldots,n\} \quad \xi(vu_l)\neq i$.
    Now it is clear that $\xi' = \chi *_i v_{\structure G}$, where $\chi$ is the coloring of $\structure G$ that satisfies $\chi = \mathrm{sw}_{\sigma}(\ldots(\mathrm{sw}_{\sigma} (\structure G, \xi_{|\structure G}, u_1)\ldots,u_n)$. 
    In the previous expression $\xi_{|\structure G}$ indicates the coloring $\xi$ restricted to the subgraph $\structure G$.
    Again, by Proposition \ref{proposition:F*-F} we get $(\structure G, \chi) \in \forb{\family F_{\mathrm{sw}}}$, and hence $\structure G$ is a yes-instance for $\Col(\family{F}^{*i}_{\mathrm{sw}})$ as desired.
\end{proof}
To conclude we show that $\Col(\family F^{*i}_{\mathrm{sw}})$ is NP-complete.
For this purpose, we show that it is equivalent to a problem whose family of obstructions falls within the scope of \Cref{corollary:main-cons}.
\begin{proposition}
    Let $\family F$ be the family of colored graphs obtained from $\family F^{*i}_{\mathrm{sw}}$ by removing $(\structure K_4, 1^{\structure K_4})$ and $(\structure K_4, 2^{\structure K_4})$, recall that $i^{\structure K_4}$ indicates the constant coloring whose image is $\{i\}$.
    Then the following holds:
    \begin{itemize}
        \item $\Col(\family F^{*i}_{\mathrm{sw}})$ is equivalent to $\Col(\family F)$, and
        \item $\family F$ belongs to the scope of~\Cref{thm:main}.
    \end{itemize}
\end{proposition}
\begin{proof}
    The first item follows immediately since every $(\structure G ,\chi^{\structure G})$ that is $\family F$-free $(\structure K_4, 1^{\structure K_4}) \not\to (\structure G ,\chi^{\structure G})$ and $(\structure K_4, 2^{\structure K_4}) \not\to (\structure G ,\chi^{\structure G})$.
    For the second: note that the $\rightarrow$-maximum of the monochromatic part is $\structure K_3$ and the minimum of the non-monochromatic part is $\structure K_4$.
\end{proof}

\end{document}